%% file: main.tex
\documentclass[sigconf]{acmart}

\usepackage{tikz}
\usepackage{amsthm}
\usepackage{amsmath}
\usepackage{bm}
\usepackage{thmtools,thm-restate}
\usepackage{enumitem}
\usepackage{pifont}
\usepackage{array}
\usepackage{float}
\usepackage{color}
\usepackage{caption}
\usepackage{subcaption}
\usepackage{multirow}
\usepackage{booktabs}
\usepackage{graphicx}
\usepackage{colortbl}
\usepackage{pythonhighlight}
\usepackage{hyperref}
\usepackage{comment}
\usepackage{lipsum}
\usepackage{bbm}
\usepackage{algorithm}
\usepackage[noend]{algpseudocode}
\usepackage{hyperref}
\usepackage{xcolor}
\usepackage{anyfontsize}
\usepackage{makecell}
\usepackage{mathtools}
\usepackage[T1]{fontenc}
\usepackage[font=small,labelfont=bf,tableposition=top]{caption}
\usepackage{balance}

\usepackage{draftwatermark}
\SetWatermarkText{Preprint}
\SetWatermarkScale{0.2}
\SetWatermarkAngle{0}
\SetWatermarkLightness{0.80}
\SetWatermarkVerCenter{0.04\paperheight}
\SetWatermarkHorCenter{0.09\paperwidth}

\hypersetup{
colorlinks=true,
linkcolor=cyan,
citecolor=red,
urlcolor=blue}
\usepackage{caption}

\newtheorem{theorem}{Theorem}
\newtheorem{lemma}[theorem]{Lemma} 
 
\newtheorem{remark}[theorem]{Remark}

\newtheorem{definition}[theorem]{Definition}

\newcommand{\pluseq}{\mathrel{{+}{=}}}

\newcolumntype{P}[1]{>{\centering\arraybackslash}p{#1}}
\newcommand{\xingzhiswallow}[1]{{}}

\renewcommand{\paragraph}[1]{\textbf{\noindent{#1}}}

\newcommand*\colourcheck[1]{%
  \expandafter\newcommand\csname #1check\endcsname{\textcolor{#1}{\ding{52}}}%
}
\newcommand*\colourxmark[1]{%
  \expandafter\newcommand\csname #1xmark\endcsname{\textcolor{#1}{\ding{55}}}%
}
\colourcheck{blue}
\colourxmark{red}
\AtBeginDocument{%
  \providecommand\BibTeX{{%
    \normalfont B\kern-0.5em{\scshape i\kern-0.25em b}\kern-0.8em\TeX}}}

\xingzhiswallow{
    \setcopyright{acmcopyright}
    \copyrightyear{2018}
    \acmYear{2018}
    \acmDOI{10.1145/1122445.1122456}
    
    \acmConference[KDD '22]{KDD '22: ACM Symposium on Neural
      Gaze Detection}{June 03--05, 2018}{Washington DC}
    \acmBooktitle{Woodstock '18: ACM Symposium on Neural Gaze Detection,
      June 03--05, 2018, Washington DC}
    \acmPrice{15.00}
    \acmISBN{978-1-4503-XXXX-X/18/06}
}
 
\copyrightyear{2022}
\acmYear{2022}
\setcopyright{acmcopyright}\acmConference[KDD '22]{Proceedings of the 28th ACM SIGKDD Conference on Knowledge Discovery and Data Mining}{August 14--18, 2022}{Washington, DC, USA}
\acmBooktitle{Proceedings of the 28th ACM SIGKDD Conference on Knowledge Discovery and Data Mining (KDD '22), August 14--18, 2022, Washington, DC, USA} \acmPrice{15.00}
\acmDOI{10.1145/3534678.3539389}
\acmISBN{978-1-4503-9385-0/22/08}

\begin{document}

%%
%% The "title" command has an optional parameter,
%% allowing the author to define a "short title" to be used in page headers.
\title[Subset Node Anomaly Tracking over Large Dynamic Graphs]{Subset Node Anomaly Tracking \\over Large Dynamic Graphs}

%%
%% The "author" command and its associated commands are used to define
%% the authors and their affiliations.
%% Of note is the shared affiliation of the first two authors, and the
%% "authornote" and "authornotemark" commands
%% used to denote shared contribution to the research.

\author{Xingzhi Guo}
% \authornote{Both authors contributed equally to this research}
%\authornotemark[1]
\email{xingzguo@cs.stonybrook.edu}
\affiliation{%
  \institution{Stony Brook University}
  \city{Stony Brook}
  \country{USA}
}

\author{Baojian Zhou}
% \footnote{Corresponding author}
\authornote{Corresponding author}
\email{bjzhou@fudan.edu.cn}
\affiliation{%
  \institution{Fudan University}
  \city{Shanghai}
  \country{China}
}

\author{Steven Skiena}
\email{skiena@cs.stonybrook.edu}
\affiliation{%
  \institution{Stony Brook University}
  \city{Stony Brook}
  \country{USA}
}

\renewcommand{\shortauthors}{Xingzhi Guo, Baojian Zhou and Steven Skiena}

\begin{abstract}

Tracking a targeted subset of nodes in an evolving graph is important for many real-world applications. Existing methods typically focus on identifying anomalous edges or finding anomaly graph snapshots in a stream way. However, edge-oriented methods cannot quantify how individual nodes change over time while others need to maintain representations of the whole graph all the time, thus computationally inefficient.

This paper proposes \textsc{DynAnom}, an efficient framework to quantify the changes and localize per-node anomalies over large dynamic weighted-graphs. Thanks to recent advances in dynamic representation learning based on Personalized PageRank, \textsc{DynAnom} is 1) \textit{efficient}: the time complexity is linear to the number of edge events and independent of node size of the input graph; 2) \textit{effective}: \textsc{DynAnom} can successfully track topological changes reflecting real-world anomaly; 3) \textit{flexible}: different type of anomaly score functions can be defined for various applications. Experiments demonstrate these properties on both benchmark graph datasets and a new large real-world dynamic graph. Specifically, an instantiation method based on \textsc{DynAnom} achieves the accuracy of 0.5425 compared with 0.2790, the best baseline, on the task of node-level anomaly localization while running 2.3 times faster than the baseline. We present a real-world case study and further demonstrate the usability of \textsc{DynAnom} for anomaly discovery over large-scale graphs.

\end{abstract}



\maketitle

\section{Introduction}
\input{fig-joe-biden}

Analyzing the evolution of events in a large-scale dynamic network is important in many real-world applications. We are motivated by the following specific scenario: 
\begin{quote}
\small
Given a person-event interaction graph containing millions of nodes with several different types (e.g., person, event, location) with a stream of weekly updates over many years, how can we identify when specific individuals significantly changed their context?
\end{quote}
The example in Fig. \ref{fig:biden-changes} shows our analysis as Joe Biden shifted his career from senator to vice president and finally to president. In each transition he relates with various intensity to other nodes across time. We seek to identify such transitions through analysis of these interaction frequencies and its inherent network structure. This task becomes challenging as the graph size and time horizon scale up. For example, the raw graph data used in Fig. \ref{fig:biden-changes} contains roughly 3.2 million nodes and 1196 weekly snapshots (each snapshot averages 4 million edges) of Person Graph from 2000 to 2022.  

Despite the extensive literature \cite{akoglu2010oddball,yu2018netwalk,bhatia2021mstream} on graph anomaly detection, previous work focuses on different problem definitions of \textit{anomaly}. On the other hand, works \cite{yoon2019fast} on graph-level anomaly detection cannot identify individual node changes but only uncover the global changes in the overall graph structure. Most representative methods leverage tensor factorization \cite{chang2021f} and graph sketching \cite{bhatia2021sketch}, detecting the sudden dis/appearance of dense subgraphs. However, we argue that since most anomalies are locally incurred by a few anomalous nodes/edges, the global graph-level anomaly signal may overlook subtle local changes. Similarly, edge-level anomaly detection \cite{eswaran2018sedanspot} cannot identify node evolution when there is no direct edge event associated with that node. When Donald Trump became the president, his wife (\textit{Melania Trump}) changed status to become First Lady, but no explicit edge changes connected her to other politicians except Trump. According to the edge-level anomaly detection, there is no evidence for Melania's status change. 

Node representation learning-based methods such as \cite{yu2018netwalk,tsitsulin2021frede} could be helpful to identify anomaly change locally. However, it is impractical to directly apply these methods on large-scale dynamic settings due to 1) the low efficiency: re/training all node embeddings for snapshots is prohibitively expensive; 2) the missing alignment: node representations of each snapshot in the embedding space may not be inherently aligned, making the anomaly calculation difficult.

Inspired by the recent advances in local node representation methods \cite{postuavaru2020instantembedding,xingzhi2021subset} for both static and dynamic graphs, we could efficiently calculate the node representations based on approximated Personalized PageRank vectors. These dynamic node representations are keys for capturing the evolution of a dynamic graph. One important observation is that the time complexity of calculating approximate PPV is $\mathcal{O}(\frac{1}{\alpha \epsilon})$  for per queried node, where $\alpha$ is PageRank \cite{page1999pagerank} teleport probability and $\epsilon$ is the approximation error tolerance. As the graph evolves, the per-node update complexity is $\mathcal{O}(\frac{|\Delta E|}{\epsilon})$, where $|\Delta E|$ is the total edge events. The other key observation is that the representation space is inherently consistent and aligned with each other, thus there is a meaningful metric space defined directly over the Euclidean space of every node representation across different times. Current local node representation method \cite{xingzhi2021subset} only captures the node-level structural changes over unweighted graphs, which ignores the important interaction frequency information over the course of graph evolution. For example, in a communication graph, the inter-node interaction frequency is a strong signal, but the unweighted setting ignores such crucial feature. This undesirable characteristic makes the node representation becomes less expressive in weighted-graph. 

%Based on the above observations, 
In this paper, we generalize the local node representation method so that it is more expressive and supports dynamic weighted-graph, and the resolve the practical subset node anomaly tracking problem. We summarize our contributions as follows:
% Efficient subset node tracking algorithm over dynamic weighted graph with theoretical analysis: 
\vspace{-1mm}
\begin{itemize}
\item We generalize dynamic forward push algorithm \cite{zhang2016approximate} for calculating Personalized PageRank, making it suitable for weighted-graph anomaly tracking. The algorithm updates node representations per-edge event (i.e., weighted edge addition/deletion), and the per-node time complexity linear to edge events ($|\Delta E|$) and error parameter($\epsilon$) which is independent of the node size.
\item We propose an efficient anomaly framework, namely \textsc{DynAnom} that can support both node and graph-level anomaly tracking, effectively localizing the period when a specified node significantly changed its context. Experiments demonstrate its superior performance over other baselines with the accuracy of 0.5425 compared with 0.2790, the baseline method meanwhile 2.3 times faster.
\item A new large-scale real-world graph, as new resources for anomaly tracking: PERSON graph is constructed as a more interpretable resource, bringing more research opportunities for both algorithm benchmarking and real-world knowledge discovery. Furthermore, we conduct a case study of it, successfully capturing the career shifts of three public figures in U.S., showing the usability of \textsc{DynAnom}.
\end{itemize}
\vspace{-2mm}

The rest paper is organized as follows: 
Section \ref{sec:related-work} reviews previous graph anomaly tracking methods. Section \ref{sec:notation-prelim}-\ref{sec:problem-formulation} describes the notation, preliminaries and problem formulation. 
We present our proposed framework -- \textsc{DynAnom} in Section \ref{sec:methods}, and  discuss experimental results in Section \ref{sec:experiments} 
%presents experiment results on benchmark graphs and a case study on the real-world graph. 
Finally, we conclude and discuss future directions in Section \ref{sec:conclusion}.
Our code and created datasets are accessible at \textcolor{blue}{\url{https://github.com/zjlxgxz/DynAnom}}. 
% Source code and datasets will be released upon publication, and included in supplementary file 

%(compiled datasets are in \textcolor{blue}{\url{https://bit.ly/3rAshBn}}).

\vspace{-2mm}
\section{Related Work}
\label{sec:related-work}
\vspace{-1mm}
We review three most common graph anomaly tasks over dynamic graphs, namely graph, edge and node-level anomaly.

\paragraph{Graph-level and edge-level anomaly.} Graph anomaly refers to sudden changes in the graph structure during its evolution, measuring the difference between consecutive graph snapshots \cite{aggarwal2011outlier, beutel2013copycatch,eswaran2018sedanspot, eswaran2018spotlight}. \citet{aggarwal2011outlier} use a structural connectivity model to identify anomalous graph snapshots. \citet{shin2017densealert} apply incremental tensor factorization to spot dense connection formed in a short time.  \citet{yoon2019fast} use the first and second derivatives of the global PageRank Vectors as the anomaly signal, assuming that a normal graph stream will evolve smoothly. On the other hand, edge anomaly identifies unexpected edges as the graph evolves where anomalous edges adversarially link nodes in two sparsely connected components or abnormal edge weight changes \cite{ranshous2016scalable,eswaran2018sedanspot,bhatia2020midas}. Specifically, \citet{eswaran2018sedanspot} propose to use the approximated node-pair Personalized PageRank (PPR) score before and after the new edges being inserted. Most recently, \citet{chang2021f} estimate the interaction frequency between nodes, and incorporate the network structure into the parameterization of the frequency function. However, these methods cannot  reveal the node local anomalous changes, and cannot identify the individual node changes for those without direct edge modification as we illustrated in introduction.

\paragraph{Node-level local anomaly.} Node anomaly measures the sudden changes in individual node's connectivity, activity frequency or community shifts \cite{wang2015localizing, yu2018netwalk}. \citet{wang2015localizing} uses hand-crafted features (e.g., node degree, centrality) which involve manual feature engineering processes. Recently, the dynamic node representation learning methods \cite{yu2018netwalk,goyal2018dyngem, zhou2018dynamic,kumar2019predicting, lu2019temporal,sankar2020dysat} were proposed. For example, a general strategy of them \cite{yu2018netwalk,kumar2019predicting} for adopting such dynamic embedding methods for anomaly detection is to incrementally update node representations and apply anomaly detection algorithms over the latent space. To measure the node changes over time, a comparable or aligned node representation is critical. \citet{yu2018netwalk} uses auto-encoder and incremental random walk to update the node representations, then apply clustering algorithm to detect the node outliers in each snapshot. However, its disadvantage is that the embedding space is not aligned, making the algorithm only detects anomalies in each static snapshot. Even worse, it is inapplicable to large-scale graph because the fixed neural models is hard to expand without retraining. Similar approaches \cite{kumar2019predicting} with more complicated deep learning structure were also studied. These existing methods are not suitable for the subset node anomaly tracking. Instead, our framework is inspired from recent advances in efficient local node representation algorithms \cite{zhang2016approximate,guo2017parallel,postuavaru2020instantembedding,xingzhi2021subset}, which is successful in handling subset node representations over large-scale graph. 

\vspace{-1.3mm}\section{Notations and Preliminaries}
\label{sec:notation-prelim}
\vspace{-1mm}
\paragraph{Notations.} Let $\mathcal{G}(\mathcal{V},\mathcal{E},\mathcal{W})$ be a directed weighted-graph where $\mathcal{V}$ is the set of nodes, $\mathcal{E}\subseteq \mathcal{V}\times \mathcal{V}$ is the set of edges, and $\mathcal{W}$ is corresponding edge weights of $\mathcal{E}$. For each node $v$, $\operatorname{Nei}(v)$ stands for out-neighbors of $v$. For all $v\in \mathcal{V}$, the generalized out-degree vector of $\mathcal{V}$ is $\bm d$ where $v$-th entry $d(v) \triangleq \sum_{u \in \operatorname{Nei}(v)}w_{(v,u)}$ and $w_{(v,u)}$ could be the weight of edge $(v,u)$ in graph $\mathcal{G}$. The generalized degree matrix is then denoted as $\bm D := \operatorname{diag}(\bm d)$ and the adjacency matrix is written as $\bm A$ with $\bm A(u,v) = w_{(u,v)}$.

\vspace{-1.9mm}
\subsection{PPR and Its Calculations}
\label{sec:ppr}
\vspace{-1mm}
As a measure of the relative importance among nodes, PPR, a generalized version of original PageRank \cite{page1999pagerank}, plays an important role in many graph mining tasks including tasks of anomaly tracking. Our framework is built on PPR. Specifically, the Personalized PageRank vector (PPV) of a node $s$ in $\mathcal{G}$ is defined as the following:

\begin{definition}[Personalized PageRank Vector
(PPV)] Given a graph $\mathcal{G}(\mathcal{V}, \mathcal{E},\mathcal{W})$ with $|\mathcal{V}| = n $ and $ |\mathcal{E}| = m$. Define the lazy random walk transition matrix $\bm P \triangleq (1-\beta) \bm D^{-1} \bm A + \beta  \bm I_n$, $\beta \in [0,1)$ where $\bm D$ is the generalized degree matrix and $\bm A$ is the adjacency matrix of $\mathcal{G}$.\footnote{Our PPV is defined based on the lazy random walk transition matrix. This definition is equivalent to the one using $\bm D^{-1}\bm A$ but with a different parameter $\alpha$. Without loss of generality, we use $\beta=0$ throughout the paper.} 
Given teleport probability $\alpha \in [0,1)$ and the source node $s$, the Personalized PageRank vector of $s$ is defined as:
\begin{equation}
\bm \pi_{\alpha, s} = (1-\alpha) \bm P^\top \bm \pi_{\alpha, s}  + \alpha \bm 1_s, \label{equ:static-ppr-static-graph}
\end{equation}
where the teleport probability $\alpha$ is a small constant (e.g. $\alpha=.15$), and $\bm 1_s$ is an indicator vector of node $s$, that is, $s$-th entry is 1, 0 otherwise. We simply denote PPV of $s$ as $\bm \pi_s$.
\label{def:static-ppr-static-graph}
\end{definition}

Clearly, $\bm \pi_s$ can be calculated in a closed form, i.e., $\bm \pi_{s} = \alpha  (\bm I_n - (1-\alpha ) \bm P^\top)^{-1} \bm 1_s$ but with time complexity $\mathcal{O}(n^3)$. The most commonly used method is \textit{Power Iteration} \cite{page1999pagerank}, which approximates $\bm \pi_{s}$ iteratively: $\bm \pi_{s}^{(t)} = (1-\alpha)  \bm P^\top \bm \pi_{s}^{(t-1)} + \alpha \bm 1_s$. After $t =\lceil \log_{1-\alpha}\epsilon\rceil$ iterations, one can achieve $\| \bm \pi_s - \bm \pi_s^{(t)} \|_1 \leq \epsilon$. Hence, the overall time complexity of power iteration is $\mathcal{O}(m \log_{1-\alpha} \epsilon)$ with $\mathcal{O}(m)$ memory requirement. However, the per-iteration of the power iteration needs to access the whole graph which is time-consuming. Even worse, it is unclear how one can efficiently use power-iteration to obtain an updated $\bm \pi_s$ from $\mathcal{G}_t$ to $\mathcal{G}_{t+1}$. Other types of PPR can be found in \citep{wang2017fora,wei2018topppr,wu2021unifying} and references therein.

\subsection{Dynamic Forward Push Algorithm}

Instead, \textit{forward push algorithm} \cite{andersen2006local}, a.k.a. the bookmark-coloring algorithm \cite{berkhin2006bookmark}, approximates $\pi_{s}(v)$ locally via an approximate $p_s(v)$. The key idea is to maintain solution vector $\bm p_s$ and a residual vector $\bm r_s$ (at the initial $\bm r_s = \bm 1_s, \bm p_s = \bm 0$). When the graph updates from $\mathcal{G}_t$ to $\mathcal{G}_{t+1}$, a variant forward push algorithm \cite{zhang2016approximate} dynamically maintains $\bm r_s$ and $\bm p_s$. We generalize it to a weighted version of \textit{dynamic forward push} to support dynamic updates on dynamic weighted-graph as presented in Algo. \ref{algo:forward-local-push-weighted}. At each local push iteration, it pushes large residuals to neighbors whenever these residuals are significant ($ |r_s(u)| > \epsilon d(u)$). Compare with power-iteration, this operation avoids the access of whole graph hence speedup the calculations. Based on \cite{zhang2016approximate}, we consider a more general setting where the graph could be weighted-graph. Fortunately, this weighted version of \textsc{DynamicForwardPush} still has \textit{invariant property} as presented in the following lemma.

\input{forward-push-algo}

\begin{lemma} [PPR Invariant Property \cite{zhang2016approximate}] Suppose $\bm \pi_s$ is the PPV of node $s$ on graph $\mathcal{G}_t$. Let $\bm p_s$ and $\bm \pi_s$ be returned by the weighted version of \textsc{DynamicForwardPush} presented in Algo. \ref{algo:forward-local-push-weighted}. Then, we have the following invariant property.
\begin{align}
\pi_s(u) &= p_s(u) + \sum_{x \in \mathcal{V}}r_s(x) \pi_s(x),  \text{ for all } u \in \mathcal{V},  \allowdisplaybreaks\\
p_s(u) + \alpha r_s(u) &= (1-\alpha) \sum_{x \in N^{in}(u)} \frac{w_{(u,x)}p_s(x)}{d(x)} + \alpha 1_{u=s},\allowdisplaybreaks
\end{align}
where $N^{in}(\cdot)$ is the in-neighbors, and $1_{u=s} = 1$ if $u = s$, 0 otherwise.
\label{lemma:ppr-invar}
\end{lemma}
\paragraph{From PPVs to node representations.} Obtained PPVs are not directly applicable to our anomaly tracking problem as operations on these $n$-dimensional vectors is time-consuming. To avoid this, we treat PPVs as the intriguing similarity matrix and project them into low-dimensional spaces using transformations such as SVD \cite{qiu2019netsmf}, random projection \cite{zhang2018billion, chen2019fast}, matrix sketch\cite{tsitsulin2021frede}, or hashing \cite{postuavaru2020instantembedding} so that the original node similarity is well preserved by the low dimension node representations. We detail this in our main section.

\section{Problem formulation}
\label{sec:problem-formulation}
Before we present the subset node anomaly tracking problem. We first define the edge events in the dynamic graph model as the following: A set of edge events $\Delta E_t$ from $t$ to $t+1$ is a set of operations on edges, i.e., edge insertion, deletion or weight adjustment while the graph is updating from $\mathcal{G}_t$ to $\mathcal{G}_{t+1}$. Mathematically, $\Delta E_t \triangleq \{ (u_0, v_0, \Delta w_{(u_0, v_0)}, (u_1, v_1, \Delta w_{(u_1, v_1)})$, $\ldots, (u_i, v_i, \Delta w_{(u_i, v_i)})\}$ where each $\Delta w_{(u_i, v_i)}$ represents insertion (or weight increment) if it is positive, or deletion (decrement) if it is negative. Therefore, our dynamic graph model can be defined as a sequence of edge events. We state our formal problem definition as the following.

\begin{definition}[Subset node Anomaly Tracking Problem]
Given a dynamic weighted-graph $\mathcal{G}_t = \langle \mathcal{V}_t,\mathcal{E}_t, \mathcal{W}_t \rangle, \forall t \in [0,T]$, consisting of initial snapshot $\mathcal{G}_0$ and the following T snapshots that have edge events $\Delta E_t, |\Delta E_t| \geq 0$ (e.g., edge addition, deletion, weight adjustment from $t-1$ to $t$). We want to quantify the status change of a node $v$ in the targeted node subset $v \in \mathcal{S} \triangleq \{ v_0, v_1, ..., v_k \} $ from $\mathcal{G}_{t-1}$ to $\mathcal{G}_{t}$ with the anomaly measure function $f(\cdot, \cdot)$ so that the measurement is consistent to human judgement, or reflects the ground-truth of node status change if available. To illustrate this edge event process, Table \ref{tab:node-track-def} presents the process for better understanding.
\label{def:node-track-def}
\end{definition}

\begin{table}[!ht]
\caption{Illustration of subset node anomaly tracking problem. At each time $t$, node representations $\bm x_i^t$ are provided, and anomaly of node changes is quantified by $f(\bm x_i^{t-1}, \bm x_i^t$), which is expected to correlated with the anomaly label if available.}
\vspace{-3mm}
\centering
\begin{tabular}{p{2.2cm}||c|c|c|c|c}
\toprule
 Timestamp & 0 & 1 & 2 & 3 & ...  \\
\hline 
 Edge events & - & $\Delta E_1$ & $\Delta E_2$ & $\Delta  E_3$ & ...  \\
\hline 
 Snapshots & $G_0$ & $G_1$ & $G_2$ & $G_3$ & ... \\
\hline 
 Score for $v_i \in \mathcal{S}$ & - & \small{$f(x_i^{0},x_i^{1})$} &  \small{$f(x_i^{1},x_i^{2})$}&  \small{$f(x_i^{2},x_i^{3})$} & ... \\
\hline 
 Label for $v_i \in \mathcal{S}$ & - & Normal &  Normal & Abnormal & ... \\
\bottomrule
\end{tabular}
\label{tab:node-track-def}
\end{table}

\vspace{-3mm}
\section{The Proposed Framework}
\label{sec:methods}
To localize the node anomaly across time, our idea is to incrementally obtain the PPVs as the node representation at each time, then design anomaly score function $f(\cdot, \cdot)$ to quantify the PPVs changes from time to time. This section presents our proposed framework $\textsc{DynAnom}$ which has three components: 1) A provable dynamic PPV update strategy extending \cite{zhang2016approximate} to weighted edge insertion/deletion; 2) Node level anomaly localization based on incremental PPVs;  3) Graph level anomaly detection based on approximation of global PageRank changes. 
We first present how local computation of PPVs can be generalized to dynamic weighted graphs, then present the unified framework for node/graph-level anomaly tracking, finally we analyze the complexity of an instantiation algorithm.
 
\subsection{Maintenance of Dynamic PPVs}

Multi-edges record the interaction frequency among nodes, and reflect the evolution of the network not only in topological structure, but also in the strength of communities. Previous works \cite{zhang2016approximate,xingzhi2021subset} focus only on the structural changes over unweighted graph, ignoring the multi-edge cases. In a communication graph (e.g., Email graph), the disadvantage of such method is obvious: as more interactions/edges are observed, the graph becomes denser, or even turn out to be fully-connected in an extreme case. Afterwards, all the upcoming interactions will not change the node status since the ignorance of interaction frequency. This aforementioned scenario motivates us to generalize dynamic PPV maintenance to weighted graph, expanding its usability for more generic graphs and diverse applications. In order to incorporate edge weights into the dynamic forward push algorithm to obtain PPVs, we modify the original algorithm \cite{andersen2006local,zhang2016approximate} by adding a weighted push mechanism as presented in Algo. \ref{algo:forward-local-push-weighted}. Specifically, for a specific node $v$, at each iteration of the push operation, we update its neighbor residual $r_s(v)$ as follows:
\begin{equation}
r_s(v) \pluseq \frac{(1-\alpha) r_s(u) w_{(u,v)}}{\sum_{v\in \operatorname{Nei}(u)} w_{(u,v)}}. \label{equ:update-rule}
\end{equation}
The modified update of $r_s(v)$ in Equ. \eqref{equ:update-rule} efficiently approximate the PPR over static weighted graph with same time complexity as the original one.  At time snapshot $\mathcal{G}_t$, the weight of each edge event $\Delta w_{(u, v)}$ could be either positive or negative, representing edge insertion or deletion) at time $t$ so that the target graph $\mathcal{G}_t$ updates to $\mathcal{G}_{t+1}$. However, this set of edge events will break up the invariant property shown in Lemma \ref{lemma:ppr-invar}. To maintenance the invariant property on weighted graphs, we present the following Thm. \ref{theorem:dyn-adjust-weighted-graph}, a generalized version from \cite{zhang2016approximate}. The key idea is to update $\bm p_s, \bm r_s$ so that updated weights satisfy the invariant property. We present the theorem as follows:

\begin{figure}[ht]
\centering
\includegraphics[width=0.47\textwidth]{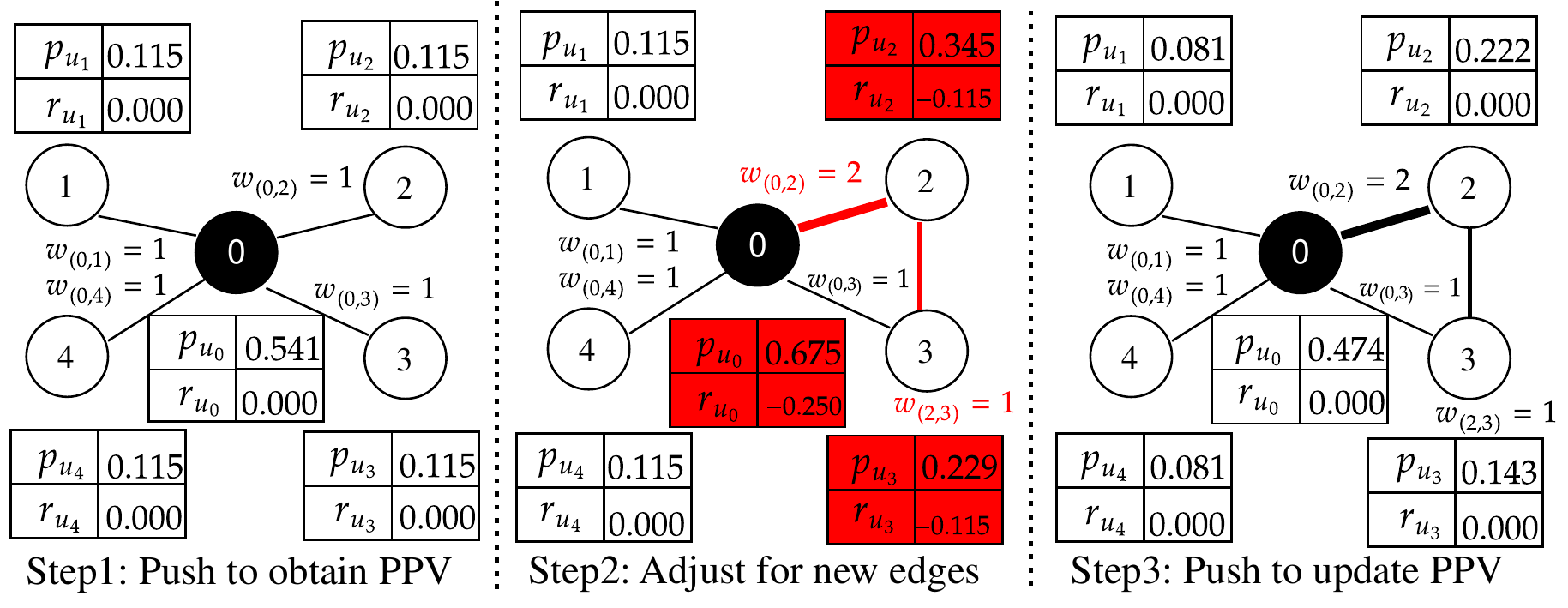}
\vspace{-2mm}
\caption{Illustration of the PPR adjustment over a  weighted graph of five nodes. \textit{Step1}: Apply Algo.\ref{algo:forward-local-push-weighted} to calculate the initial $p_s,r_s$ where $s=v_0, \alpha = 0.15$ over the initial graph. \textit{Step2}: After inserting new edges  $(v_0,v_2,1), ((v_2,v_3,1))$, the strength between $v_0$ and $v_2$ increases, and $v_2$ builds a new connection to $v_3$, which both break the invariant in Lemma \ref{lemma:ppr-invar}. We adjust to recover the invariant by applying Theorem \ref{theorem:dyn-adjust-weighted-graph} in the red-colored blocks. \textit{Step3:} We re-apply Algo.\ref{algo:forward-local-push-weighted} and update $p_s$ for better approximation error. %The updated $p_s$ is same to the PPV calculated over the updated graph snapshot. 
}
\label{fig:new-edge-example}
\end{figure}

\paragraph{Generalized PPR update rules.} Figure \ref{fig:new-edge-example} illustrates the dynamic maintenance of PPV over a weighted graph where two new edges are added, incurring edge weight increment between node 0 and 2, and new connection between node 2 and 3. Note that one only needs to have $\mathcal{O}(1)$ time to update $p_s(u), r_s(u)$, and $r_s(v)$ for each edge event. As proved, the total cost of $m$ edge events will be linearly dependent on, i.e., $\mathcal{O}(m)$. This efficient update procedure is much better than naive updates where one needs to calculate all quantities from scratch.

\begin{theorem} [Generalized PPR update rules]
Given an edge event $(u,v, \Delta w_{(u,v)})$. We use the the following update rules
\begin{align}
p'_s(u) &= p_s(u) \frac{\sum_{v \in \operatorname{Nei}(u)} w_{(u,v)}+\Delta w_{(u,v)}}{\sum_{v \in \operatorname{Nei}(u)} w_{(u,v)}}, \label{eq:p-update-weighted-body} \\
r_s'(u) &=  r_s(u) -\frac{\Delta w_{(u,v)} p_s(u)}{\alpha \sum_{v \in \operatorname{Nei}(u)} w_{(u,v)}}, \label{eq:r-update-weighted-body} \\
r_s'(v) &= r_s(v) + \frac{(1-\alpha)}{\alpha} \frac{\Delta w_{(u,v)}p_s(u)}{\sum_{v \in \operatorname{Nei}(u)} w_{(u,v)}}. \label{eq:r-v-update-weighted-body}
\end{align}
After applying a sequence of updating rules and using \textsc{DynamicForwardPush} algorithm,  then $\bm p_s$ and $\bm r_s$ satisfy invariant property in Lemma \ref{lemma:ppr-invar} and accurately approximate PPV, as proved in Appendix \ref{appendix-proof}.
\label{theorem:dyn-adjust-weighted-graph}
\end{theorem}

\paragraph{From dynamic PPVs to node representations.} The above theorem provides efficient rules to incrementally update PPVs for dynamic weighted graphs. To obtain dynamic node representation, two universal hash functions $h_{\operatorname{d}}$ and $h_{\operatorname{sgn}}$ are used (see details in \cite{postuavaru2020instantembedding}).  In the following section, we describe how to leverage PPVs as node representation for tracking node/graph-level anomaly.
 
\subsection{Anomaly Tracking Framework: \textsc{DynAnom}}

The key idea of our proposed framework is to measure the individual node anomaly in the embedding space based on Personalized PageRank \citep{postuavaru2020instantembedding,xingzhi2021subset}, and aggregate the individual anomaly scores to represent graph-level anomaly. By the fact that the PPV-based embedding space is aligned, we can directly measure the node anomaly by any distance function.\footnote{In this paper, we explored $\ell_1$-distance function where given $\bm x,\bm y$, $f(\bm x, \bm y) = \| \bm x - \bm y\|_1$.} We present our design as follows:

\subsubsection{Node-level Anomaly Tracking}
Since the inherently local characteristic of our proposed framework, we can efficiently measure the status of arbitrary nodes in two consecutive snapshots in an incrementally manner. Therefore, we could efficiently solve the problem defined in Def. \ref{def:node-track-def} and localize where the anomaly actually happens by ranking the anomalous scores. To do this, the key ingredient is to measure node difference from $t-1$ to $t$ as $\delta_s^t = f(\bm x_s^{t-1} , \bm x_s^{t})$ where $\bm x_s^{t}$ is the node representation of node $s$ at time $t$. Based on the above motivation, we first present the node-level anomaly detection algorithm \textsc{DynAnom} in Algo. \ref{algo:dynanom-node} where the score function $f(\cdot,\cdot)$ is realized by $\ell_1$ norm in our experiments. There are three key steps: \textit{Step 1.} calculating dynamic PPVs for  all nodes in $S$ over $T$ snapshots (Line 3); \textit{Step 2.} obtaining node representations of $S$ by local hash functions. Two universal hash functions\footnote{We use \textit{sklearn.utils.murmurhash3\_32} in our implementation.} are used in \textsc{DynNodeRep} (Line 5); \textit{Step 3.} calculating node anomaly scores based on node representation $\bm x_s$ for all $s \in S$.

Noted that our design directly use the first-order difference ($\nabla x_s$) between two consecutive time. Although there are more complex design of $f(\cdot, \cdot)$ under this framework, such as, second-order measurement $\nabla^2 x_s$ similar as in \citet{yoon2019fast}, clustering based anomaly detection \cite{yu2018netwalk}, or supervised anomaly learning. We restrict our attention on the simplest form by setting $p=1 \text{ or } 2$ with the first-order measure. In our framework, the definition of anomaly score can be highly flexible for various applications.

\begin{algorithm}[ht]
\caption{$\textsc{DynAnom}(\mathcal{G}_0, \Delta E_{1, \ldots, T} , \mathcal{S}, \epsilon, \alpha, p)$ }
\begin{algorithmic}[1]
\State \textbf{Input: }  Initial graph $\mathcal{G}_0$, Edge events $ \Delta E_{1, \ldots, T}$ , Subset target nodes $\mathcal{S}$,  PPV quality $\epsilon$,  teleport factor $\alpha $.
\State \textcolor{blue}{//Step 1: Incrementally calculate PPVs}
\State $\bm p_{s\in \mathcal{S}} = \textsc{IncrementPush}(\mathcal{G}_0, \Delta E_{1, \ldots, T} , S, \epsilon, \alpha)$
\State \textcolor{blue}{//Step 2: Obtain node representations in target set $\mathcal{S}$ }
\State $\bm x_{s\in \mathcal{S}} = \textsc{DynNodeRep}(p_s^0, ... ,P_s^T)$
\State \textcolor{blue}{//Step 3: Calculate node anomaly for all nodes in $\mathcal{S}$}
\For{$s \in S, t \in [1,T] $}
    \State $\delta_s^t = f( x_s^{t-1}, x_s^{t}) = \left( \sum_{i} \left| x_s^{t}(i) - x_s^{t-1}(i) \right|^p \right)^ {1/p}$
\EndFor
\State \textbf{return} $\bm \delta_s = [\delta_s^1,\delta_s^2,\ldots, \delta_s^T], \forall s \in S$

\noindent\hrulefill
\Procedure{IncrementPush}{$\mathcal{G}_0, \Delta E_{1, \ldots, T} , \mathcal{S}, \epsilon, \alpha$}
\State $t=0$
\For{$s \in \mathcal{S} := \{v_1,v_2,\ldots, v_k\}$}
\State $p_s^t = \bm 0, \quad r_s^t = 1_s $
\State $p_s^t, r_s^t=$\textsc{DynamicForwardPush}$(\mathcal{G}^0, s, p_s^t, r_s^t, \epsilon, \alpha)$
\EndFor
\For{$t \in [1,T]$}
    \For{$s \in \mathcal{S}$}
        \For{$(u,v,\Delta w_{(u,v)}) \in \Delta E_t$}
            \State update $p_s^t(u), r_s^t(u), r_s^t(v)$ uses rules in Thm. \ref{theorem:dyn-adjust-weighted-graph}.
        \EndFor
    \State $p_s^t, r_s^t=$\textsc{DynamicForwardPush}$(\mathcal{G}^0, s, p_s^t, r_s^t, \epsilon, \alpha)$
    \EndFor
    \State $\mathcal{G}_t += (u,v,\Delta w_{(u,v)})$
\EndFor
\State \textbf{return} $ \bm p_s =[ p_s^0, ... , p_s^T ], \forall s \in S$ 
\EndProcedure
\noindent\hrulefill
\Procedure{DynNodeRep}{$\bm p_s^0, \ldots ,\bm p_s^T$}
\State $\epsilon_{c}= \textsc{min}(\frac{1}{|\mathcal{V}|}, \text{1e-5}), dim = 1024$
\For{$t \in [1,T)$}
    \For{$i \in \cup_{t'\in\{t, t-1\}} \textsc{supp}(p_s^{t'})$}:
    \State $p_s^{t-1}(i) = 0$ if $p_s^{t-1}(i) \leq \epsilon_c $
    \State $p_s^t(i) = 0$ if $p_s^t(i) \leq \epsilon_c $
    \EndFor
    \State $x_s^{t} = \textsc{ReduceDim}(\frac{p_s^t}{\|p_s^t\|_1}, dim)$
    \State $x_s^{t-1} = \textsc{ReduceDim}( \frac{p_s^{t-1}}{\|p_s^{t-1}\|_1}, dim)$
    \State \textbf{return} $ \bm x_s =[ x_s^0, ... , x_s^T ], \forall s \in S$ 
\EndFor
\EndProcedure
\noindent\hrulefill
\Procedure{ReduceDim}{$x,dim$}
\State // Hash function $h_{dim}(i): \mathbb{N}\to [dim]$
\State // Hash function $h_{\operatorname{sgn}}(i): \mathbb{N}\to \{ \pm 1\}$
\If{$\textsc{dim}(x) \leq dim  $}
    \State \textbf{return} $x$
\Else
    \State $\bar x = \bm 0 \in \mathcal{R}^{dim}$
    \For{$i \in \textsc{Supp}(x)$}
        \State $\bar x(h_{dim}(i)) \pluseq h_{\operatorname{sgn}}(i) \log \left(x(i) \right)$
    \EndFor
    \State \textbf{return} $ \frac{\bar x}{\| \bar x \|_1}$
\EndIf
\EndProcedure
\end{algorithmic}
\label{algo:dynanom-node}
\end{algorithm}

\subsubsection{Graph-level Anomaly Tracking}

Based on the node-level anomaly score proposed in Algo. \ref{algo:dynanom-node}, we propose the graph-level anomaly. The key property we used for graph-level anomaly tracking is the linear relation property of PPV. More specifically, let $\bm \pi$ be the PageRank vector. We note that the PPR vector is linear to personalized vector $\bm \pi_s$. Therefore, global PageRank could be calculated by the weighted average of single source PPVs.
\vspace{-1.0mm}
\begin{align}
\bm \pi = \alpha \frac{\bm \vec 1}{|\mathcal{V}|} \sum_{i = 0}^{\infty}(1-\alpha)^{i} \bm P^{i}  = \sum_{ s \in \mathcal{V} } \frac{1_s}{|\mathcal{V}|} \bm \pi_{s} = \frac{1}{|\mathcal{V}|} \sum_{ s \in \mathcal{V} } \bm \pi_{s}, \allowdisplaybreaks
\end{align}
From the above linear equality, we formulate the global PPV, denoted as $\pi_{, g}$,as the linear combination of single-source PPVs as following:
\vspace{-1.0mm}
\begin{align}
\bm \pi_{g} = \sum_{ s \in \mathcal{V} } \gamma_s \pi_{s} , \gamma_s = \frac{d(s)}{m}, m= \sum_{i \in \mathcal{V}} d(i) \allowdisplaybreaks
\end{align}

In order to capture the graph-level anomaly, we use the similar heuristic weights \cite{yoon2019fast} $\gamma_s = \frac{d(s)}{m}$, which implies that $\bm \pi_{\alpha}$ is dominated by high degree nodes, thus the anomaly score is also dominated by the greatest status changes from high degree nodes as shown below:
\vspace{-1.0mm}
\begin{align}
    \| \bm \pi_{g}^t - \bm \pi_{g}^{t-1} \|_{1} & = 
    \| \sum_{ s \in \mathcal{V} } \gamma_s^{t-1} \bm \pi^{t-1}_{ s}  - \sum_{ s \in \mathcal{V} } \gamma_s^t \bm \pi^t_{ s} \|_1  \nonumber \allowdisplaybreaks\\
    &= \| \sum_{ s \in \mathcal{V} } \frac{d_s^{t-1}}{m^{t-1}} \bm \pi^{t-1}_{ s}  - \sum_{ s \in \mathcal{V} } \frac{d_s^{t}}{m^{t}} \bm \pi^t_{ s} \|_1 \nonumber \allowdisplaybreaks\\
    &\approx   \sum_{ s \in \mathcal{V}^t_{high}}  \frac{d_s^{t}}{m^{t}} \| \bm \pi^{t-1}_{ s}  - \bm \pi^t_{ s} \|_1 \nonumber \allowdisplaybreaks\\
    &\leq \sum_{ s \in \mathcal{V}^t_{high}}  \| \bm \pi^{t-1}_{ s}  - \bm \pi^t_{ s} \|_1,\allowdisplaybreaks
\end{align}
where $\mathcal{V}_{high}^t$ denotes the set of high-degree nodes. Note that the approximate upper bound of $\| \bm \pi_g^t - \bm \pi_g^{t-1}\|_1$ can be lower bounded :
\begin{align}
\footnotesize
\sum_{ s \in \mathcal{V}^t_{high}}  \| \bm \pi^{t-1}_{ s}  - \bm \pi^t_{ s} \|_1 \geq   \textsc{max}(\{ \left\| \bm \pi^{t-1}_{ s} - \bm \pi^t_{ s} \right\|_1 ,\forall s \in \mathcal{V}^t_{high} \} ).
    \label{eq:graph-anomaly-score}
\end{align}
By the approximation of global PPV difference in Equ. \eqref{eq:graph-anomaly-score}, we discover that the changes of high degree nodes will greatly affect $\ell_1$-distance, and similarly for $\ell_2$-distance. The above analysis assumes $d_s^{t-1} / m^{t-1} \approx d_s^{t} / m^{t}$.  We can also hypothesize that the changes in high-degree node will flag graph-level anomaly signal. Therefore, we use high-degree node tracking strategy for graph-level anomaly:
\begin{align}
\textsc{DynAnom}_{graph}^t &:= \textsc{max}(\{ \| \bm \pi^{t-1}_{ s} - \bm \pi^t_{ s} \|_1 ,\forall s \in \mathcal{V}_{high} \} ). \allowdisplaybreaks
\end{align}
In practice, we incrementally add high degree nodes (e.g., top-50) into the tracking list from each snapshots and extract the maximum PPV changes among them as the graph-level anomaly score. This proposed score could capture the anomalous scenario where one of the top nodes (e.g., popular hubs) encounter great structural changes, similar to the DDoS attacks on the internet.

\subsubsection{Comparison between \textsc{DynAnom} and other methods.}
The significance of our proposed framework compared with other anomaly tracking methods is illustrated in Table \ref{tab:algo-comparison}.  \textsc{DynAnom} is capable of detecting both node and graph-level anomaly, supporting all types of edge modifications, efficiently updating node representations (independent of $|\mathcal{V}|$), and working with flexible anomaly score functions customized for various downstream applications.

\newcolumntype{D}[1]{>{\centering\let\newline\\\arraybackslash\hspace{0pt}}m{#1}}

\begin{table}[ht]
        \centering
\caption{The supported features of DynAnom and other baselines.  }
\small
\vspace{-2mm}
\begin{tabular}{@{\extracolsep{2pt}}p{0.13\linewidth}p{0.04\linewidth}p{0.04\linewidth}p{0.06\linewidth}p{0.06\linewidth}p{0.09\linewidth}p{0.05\linewidth}p{0.06\linewidth}p{0.06\linewidth}@{}}
\toprule 
  \multirow{2}{*}{\diaghead(4,-3){\hskip 13mm} {Method}{Feature}}  & \multicolumn{2}{D{0.14\linewidth}}{Anomaly} & \multicolumn{3}{D{0.25\linewidth}}{Edge Event Types} & \multicolumn{3}{D{0.24\linewidth}}{ Algorithm  } \\
\cmidrule{2-3} \cmidrule{4-6}\cmidrule{7-9}

& Node level & Graph level & Edge Stream & Add/ Delete & Weight Adjust & $|\mathcal{V}|$ Ind. & Align Repr.  & Flex. Score  \\
\midrule 
 SedanSpot & \redxmark  & \redxmark & \bluecheck & \redxmark & \redxmark & \bluecheck & \redxmark & \redxmark \\
 AnomRank & \redxmark & \bluecheck & \redxmark & \bluecheck & \bluecheck & \redxmark & \redxmark & \redxmark \\
 NetWalk & \bluecheck & \redxmark & \redxmark & \bluecheck & \redxmark & \redxmark & \redxmark & \redxmark \\
 DynPPE & \bluecheck & \redxmark & \bluecheck & \bluecheck & \redxmark & \bluecheck & \bluecheck & \redxmark \\
 \hline
 \textbf{DynAnom} & \bluecheck & \bluecheck & \bluecheck & \bluecheck & \bluecheck & \bluecheck & \bluecheck & \bluecheck \\
 \bottomrule
\end{tabular}
\label{tab:algo-comparison}
\end{table}

\paragraph{Time Complexity Analysis.}
The overall complexity of tracking subset of $k$ nodes across $T$ snapshots depends on run time of three main steps as shown in Algo. \ref{algo:dynanom-node}: 1. the run time of \textsc{IncrementPush} for nodes $\mathcal{S}$; 2. the calculation of dynamic node representations, which be finished in $\mathcal{O}(k T |\operatorname{supp}(\bm p_{s^\prime})|)$ where $|\operatorname{supp}(\bm p_{s^\prime})|$ is the maximal support of all $k$ sparse vectors, i.e. $|\operatorname{supp}(\bm p_{s^\prime})| = \max_{s\in \mathcal{S}} |\operatorname{supp}(\bm p_{s})|$. The overall time complexity of our proposed framework is stated as in the following theorem.
\begin{theorem}[Time Complexity of \textsc{DynAnom}]
Given the set of edge events where $\mathbb{E} = \{ e_1, e_2,\ldots, e_m\}$ and $T$ snapshots and subset of target nodes $\mathcal{S}$ with $|\mathcal{S}| = k$, the instantiation of our proposed \textsc{DynAnom} detects whether there exist events in these $T$ snapshots in $\mathcal{O}( k m/\alpha^2 + k \bar{d^t} / (\epsilon \alpha^2) +  k/(\epsilon \alpha) + k T s)$ where $\bar{d^t}$ is the average node degree of current graph $\mathcal{G}_t$ and $s\triangleq \max_{s\in\mathcal{S}} |\operatorname{supp}(\bm p_s)|$.
\label{thm:time-complexity}
\end{theorem}

Note that the above time complexity is dominated by the time complexity of \textsc{IncrementPush} following from \cite{xingzhi2021subset}. Hence, the overall time complexity is linear to the size of edge events. Notice that in practice, we usually have $s \ll n$ due to the sparsity we controlled in \textsc{DynNodeRep}.

\section{Experiments}
\label{sec:experiments}

In this section, we demonstrate the effectiveness and efficiency of \textsc{DynAnom} for both node-level and graph-level anomaly tasks over three state-of-the-art baselines, followed by one anomaly visualization of ENRON benchmark, and another case study of a large-scale real-world graph about changes in person's life.
\vspace{-1.5mm}
\subsection{Datasets}
\paragraph{DARPA}: The network traffic graph \cite{lippmann20001999}. Each node is an IP address, each edge represents network traffic. Anomalous edges are manually annotated as network attack (e.g., DDoS attack). 
\paragraph{ENRON}\cite{nr-aaai15, cohen2005enron}: The email communication graph. Each node is an employee in Enron Inc. Each edge represents the email sent from one to others. There is no human-annotated anomaly edge or node. Following the common practice used in \cite{yoon2019fast, chang2021f}, we use the real-world events to align the discovered highly anomalous time period.
\paragraph{EU-CORE}: A small email communication graph for benchmarking. Since it has no annotated anomaly, we use two strategies (Type-s and Type-W), similar to \cite{eswaran2018sedanspot},  to inject anomalous edge into the graph. \textbf{EU-CORE-S} includes anomalous edges, creating star-like structure changes in 20 random snapshots. Likewise, \textbf{EU-CORE-L:} is injected with multiple edges among randomly selected pair of nodes, simulating  node-level local changes.  We describe details in appendix. \ref{sec:appendix-data}. 
\paragraph{Person Knowledge Graph (PERSON)}: We construct PERSON graph from EventKG \cite{gottschalk2019eventkg, gottschalk2018eventkg}. Each node represents a person or an event. The edges represent a person's participation history in events with timestamps. The graph has no explicit anomaly annotation. Similarly, we align the detected highly anomalous periods of a person with the real-world events, which reflects the person's significant change as his/her life events update. We summarize the statistics of dataset in the following table:

\xingzhiswallow{
\begin{table}[ht]
        \centering
        \small
\begin{tabular}{p{0.05\textwidth}|p{0.05\textwidth}p{0.07\textwidth}p{0.05\textwidth} p{0.06\textwidth} p{0.06\textwidth}}
\toprule 
%   &  \multicolumn{2}{l}{Quantitatively studied} & \multicolumn{3}{l}{Qualitatively studied }  \\
%  \hline 
 Dataset & DARPA & EU-CORE & ENRON & PERSON & COUNTRY \\
\midrule 
 $\displaystyle |V|$ & 25,525 & 986 & 151 & 609,930 & 206 \\
\hline 
 $\displaystyle |E|$ & 4,554,344 & 333,734 & 50,572 & 8,782,630 & 2,013,931 \\
\hline 
%  Duration & 87,726 seconds & 69,459,255 s & \textasciitilde 3 years & 43 years from 1980 & 43 years from 1979\\
% Duration & 87,726 seconds & 69,459,255 s & 98,284,400 s (\textasciitilde 3 years) & 43 years from 1980 & 43 years from 1979\\
% \hline 
 $|\mathcal{G}_{0,...,T}|$ & 1463 & 115 & 1138 & 43 & 43 \\
\hline 
 Init. $t$ & 256 & 25 & 256 & 20  & 11 \\
\bottomrule 
%  Initial T for $\mathcal{G}_0$ & 256 & 25 & 256 & 20 ( year 2000) & 11 (year 1990) \\
% \hline 
%  Tasks & Node and Graph Anomaly  & Node Anomaly  & Case Study &  Case Study &  Case Study \\
% \hline
\end{tabular}
\caption{Dataset Statistics. We convert them into undirected graph by repeating edges with reversed source and destination node.}
\end{table}

}

\vspace{-1.5mm}
\begin{table}[ht]
\centering
\caption{Dataset Statistics. Graphs are converted into undirected by repeating edges with reversed source and destination node. We detect anomalies after the initial snapshot (Init. $t$).}
\vspace{-3mm}
\begin{tabular}{p{0.05\textwidth}|p{0.08\textwidth}p{0.08\textwidth}p{0.08\textwidth} p{0.08\textwidth} }
\toprule 
%   &  \multicolumn{2}{l}{Quantitatively studied} & \multicolumn{3}{l}{Qualitatively studied }  \\
%  \hline 
 Dataset & DARPA & EU-CORE & ENRON & PERSON  \\
\midrule 
 $\displaystyle |V|$ & 25,525 & 986 & 151 & 609,930  \\
\hline 
 $\displaystyle |E|$ & 4,554,344 & 333,734 & 50,572 & 8,782,630 \\
\hline 
%  Duration & 87,726 seconds & 69,459,255 s & \textasciitilde 3 years & 43 years from 1980 & 43 years from 1979\\
% Duration & 87,726 seconds & 69,459,255 s & 98,284,400 s (\textasciitilde 3 years) & 43 years from 1980 & 43 years from 1979\\
% \hline 
 $|\mathcal{G}_{0,...,T}|$ & 1463 & 115 & 1138 & 43  \\
\hline 
 Init. $t$ & 256 & 25 & 256 & 20   \\
\bottomrule 
%  Initial T for $\mathcal{G}_0$ & 256 & 25 & 256 & 20 ( year 2000) & 11 (year 1990) \\
% \hline 
%  Tasks & Node and Graph Anomaly  & Node Anomaly  & Case Study &  Case Study &  Case Study \\
% \hline
\end{tabular}
\end{table}
% \vspace{-4mm}

%dataset_gdelt-country-level-m-conflict-edge-197901-202112-num-art-5-gs-scale-7.0-initSS_1-timeStep_31536000
\vspace{-2mm}
\subsection{Baseline Methods}
We compare our proposed method to four representative state-of-the-art methods over dynamic graphs:
\paragraph{Graph-level method: AnomRank}\cite{yoon2019fast} \footnote{We discovered potential bugs in the official code, which make the results different from the original paper. We have contacted authors for further clarification.} uses the global Personalized PageRank and use 1-st and 2-nd derivative to quantify the anomaly score of graph-level.  We use the node-wise AnomRank score for the node-level anomaly,
\paragraph{Edge-level method: SadenSpot} \cite{eswaran2018sedanspot} approximates node-wise PageRank changes before and after inserting edges as the edge-level anomaly score. 
\paragraph{Node-level method: NetWalk} \cite{yu2018netwalk} incrementally updates the node embeddings by extending random walks. However, since the \textit{anomaly} defined in the original paper is the outlier in each static snapshot, we re-alignment \cite{gower1975generalized, wang2008manifold} NetWalk embeddings and apply $\ell_2$-distance for node-level anomaly. \footnote{We use Procrustes analysis to find optimal scale, translation and rotation to align}
\paragraph{Node-level DynPPE} \cite{xingzhi2021subset} used the local node embeddings based on PPVs for unweighted graph. We extend its core PPV calculation algorithm for weighted graphs. The detailed hyper-parameter settings are listed in Table \ref{tab:param-config} in appendix. In the following experiments, we demonstrate that our proposed algorithm outperforms over the strong baselines.

\subsection{Exp1: Node-level Anomaly Localization}
\label{sec:exp1}

\paragraph{Experiment Settings:} As detailed in Def.\ref{def:node-track-def}, given a sequence of graph snapshots $  \mathcal{\bm G} = \{ \mathcal{G}_0, \mathcal{G}_1, ..., \mathcal{G}_T \}$ and a tracked node subset $\bm S = \{ u_0, u_1, ..., u_i \}$. 
We denote the set of anomalous snapshots of node u as ground-truth, such that $\bm Y_u$, containing $k_u$ timestamps where there is at least one new anomalous edge connected to node $u$ in that time.
We calculate the node anomaly score for each node. For example for node u: $\bm \delta_u = \{ \hat \delta_u^1, \hat \delta_u^2, \ldots, \hat \delta_u^T \}$, and rank its anomaly scores across time. 
Finally, we use the set of the top-$k_u$ highest anomalous snapshots $ \hat{\bm Y_u} $ as the predicted anomalous snapshot \footnote{In experiment, we assume that the number of anomalies is known for each node to calculate the prediction precision.} and calculate the averaged prediction precision over all nodes as the final performance measure as shown below:
%we identify at which snapshots did the node anomalies happen. We annotate node u's anomaly of when anomalous edges occurs in the node at timestamp t 
%It is a proxy for node-level anomaly, and easier for quantitative experiments.
%For example, in a graph of 5 snapshots, for the node $u_k$, it has anomalous edges in snapshot $\{2,4\}$ and non-anomalous edges at $0,1,3$, we rank the node-level anomaly score of $u_k$ for every snapshot, and take the top-2 snapshot with highest snapshots as the anomaly localization results.  
% We rank each node by its anomaly scores across time, and take the top-k snapshots $\bm \hat T_u^{k}$ as the predicted anomalous snapshot \footnote{In experiment, we assume that the number of anomalies is known for each node to calculate the prediction accuracy.}, then calculated the averaged prediction accuracy over all nodes as the final performance measure. 
\vspace{-1mm}
\begin{align}
    Precision_{avg} = \frac{1}{|\bm S|} \sum_{u \in S}\frac{ \hat{\bm Y_u} \cap \bm  Y_u  }{ |\hat{\bm Y_u}| } \nonumber
\end{align}
\vspace{-1mm}
% In order to leverage the annotated edge anomaly as the ground-truth,  we define the time (snapshot) of node-level anomaly to be the time when the anomalous edges get connected to the end-points. Note that, though this ground-truth is related to edge anomaly, we independently measure the node-level anomaly score of each node (not the edge). 
For DARPA dataset, we track 151 nodes\footnote{We track totally 200 nodes, but 49 nodes have all anomalous edges before the initial snapshot, so we exclude them in precision calculation.} which have anomalous edges after initial snapshot. Similarly we track 190,170 anomalous nodes over EU-CORE-S and EU-Core-L.
% Since EU-CORE dataset does not have annotated anomalous edges,  we randomly select 20 snapshots to inject artificial anomalous edges using two popular injection methods,  which are similar to the approach used in \cite{yoon2019fast}. For each selected snapshot in \textit{EU-CORE-S}, we uniformly select one node,$u_{high}$,  from the top 1\% high degree nodes, and injected 70 multi-edges connecting the selected node to other 10 random nodes which were not connected to $u_{high}$ before, which simulates the structural changes.  In each selected snapshot in \textit{EU-CORE-L}, we uniformly select 5 pairs of nodes, and injected edge connect each pair with totally 70 multi-edges as anomaly, which simulates the sudden peer-to-peer communication.  
We exclude edge-level method SedanSpot for node-level task because it can not calculate node-level anomaly score. We present the precision and running time in Table \ref{tab:node-anomaly-loc-accuracy} and Table \ref{tab:time-node-anomaly}.

\begin{table}[ht]
    \centering
    \caption{The average precision of node-level anomaly localization. Our proposed \textsc{DynAnom} outperforms other baselines.}
    % \small
    \vspace{-2mm}
    \begin{tabular}{lrrr}
    \toprule 
      & DARPA & EU-CORE-S  & EU-CORE-L  \\
    \midrule 
     Random  & 0.0075 & 0.0142 & 0.0088 \\
    \hline 
     AnomRank & 0.2790 & 0.2173 & 0.4019 \\
    \hline 
     AnomRankW & 0.2652 & 0.2213 & 0.4078 \\
    \hline 
     NetWalk  & OOM & 0.0421 & 0.0416 \\
    \hline 
     DynPPE &  0.1701 & 0.2478 & 0.2372 \\
    \hline 
     \textbf{DynAnom} & \textbf{0.5425} & \textbf{0.4242} & \textbf{0.5215} \\
     \bottomrule
    \end{tabular}
\label{tab:node-anomaly-loc-accuracy}
\end{table}
\vspace{-2mm}

\paragraph{Effectiveness.} In Table \ref{tab:node-anomaly-loc-accuracy}, the low scores of \textit{Random}\footnote{For Random baseline, we randomly assign anomaly score for each node across snapshots} demonstrate the task itself is very challenging. We note that our proposed \textsc{DynAnom} outperforms all other baselines, demonstrating its effectiveness for node-level anomaly tracking.

\paragraph{Scalability.} NetWalk hits Out-Of-Memory (OOM) on DARPA dataset due to the Auto-encoder (AE) structure where the parameter size is related to $|\mathcal{V}|$. Specifically, the length of input/output one-hot vector is $|\mathcal{V}|$, making it computationally prohibitive when dealing with graphs of even moderately large node set.  Moreover, since the AE structure is fixed once being trained, there is no easy way to incrementally expand neural network to handle the new nodes in dynamic graphs. While \textsc{DynAnom} can easily track with any new nodes by adding extra dimensions in $\bm p_s, \bm r_s$ and $\bm d$ and apply the same procedure in an online fashion.

\paragraph{Power of Node-level Representation.} The core advantage of \textsc{DynAnom} over AnomRank(W) is the power of node-level representation. AnomRank(W) essentially represents individual node's status by one element in the PageRank vector. While \textsc{DynAnom} uses node representation derived from PPVs, thus having far more expressivity and better capturing  node changes over time.

\paragraph{Advantage of Aligned Space.} Despite that NetWalk hits OOM on DARPA dataset, we hypothesize that the reason of its  poor performance on EU-CORE\{S,L\} is the embedding misalignment. The incremental training makes the embedding space not directly comparable from time to time even after extra alignment. Although it's  useful for classification or clustering within one snapshot, the misaligned embeddings may cause troublesome overhead for downstream tasks, such as training separate classifiers for every snapshot.

\paragraph{Importance of Edge Weights.} Over all datasets, our proposed \textsc{DynAnom} consistently outperforms DynPPE, which is considered as the unweighted counterpart of \textsc{DynAnom}. It demonstrates the validity and versatility of our proposed framework.

\vspace{-2mm}
\begin{table}[ht]
\centering
\caption{For comparing running time (in seconds), excluding the time for data loading and graph update. }
\vspace{-2mm}
\begin{tabular}{lrrr}
\toprule 
  & DARPA & EU-CORE-S & EU-CORE-L \\
\midrule 
 AnomRank(W) & 905.983  &  26.084& 26.865 \\
\hline 
 NetWalk & OOM &  3649.14 &  3644.58\\
\hline 
DynPPE &  84.247 &  33.835 & 30.933 \\
\hline 
 DynAnom & 379.334 & 30.054 & 26.359 \\
 \bottomrule
\end{tabular}
\label{tab:time-node-anomaly}
\end{table}
\vspace{-2mm}

\paragraph{Efficiency.} Table \ref{tab:time-node-anomaly} presents the wall-clock time \footnote{We track function-level running time by CProfile, and remove those time caused by data loading and dynamic graph structure updating}. Deep learning based NetWalk is the slowest. Although it can incrementally update random walk for the new edges, it has to  re-train or fine-tune for every snapshot. 
At the first glance, DynPPE achieves the shortest running time on DAPRA dataset, but it is an illusive victory caused by dropping numerous multi-edge updates at the expense of precision. Furthermore, the Power Iteration-based AnomRank(W) is slow on DARPA as we discussed in Section \ref{sec:notation-prelim}, but the speed difference may not be evident when the graph is small (e.g., DARPA has 25k nodes, and EU-CORE has less than 1k nodes). Moreover, our algorithm can independently track individual nodes as a result of its local characteristic,  it could be further speed up by embarrassingly parallelization on a computing cluster with massive cores.
% Moreover, on \textit{large} graph (DAPRA of 25k nodes), but the speed advantage may not be evident when the graph is small (less than 1k nodes).

% \paragraph{Observations from Table \ref{tab:time-node-anomaly}}: \textbf{I.} Although it seems that DynPPE achieves the shortest running time on DAPRA dataset, it is an illusive victory caused by dropping numerous multi-edges and sacrifice performance as presented in \ref{tab:node-anomaly-loc-accuracy}. 

% \textbf{II.} AnomRank(W) is based on Power Iteration, which could be slower as we discussed in Section \ref{sec:notation-prelim}. On the other hand, our local method could be faster on \textit{large} graph (DAPRA of 25k nodes), but the speed advantage may not be evident when the graph is small (less than 1k nodes).

\vspace{-1mm}
\subsection{Exp2: Graph-level Anomaly Detection }
\label{sec:exp2}
\paragraph{Experiment Setting}:
Given a sequence of graph snapshots $ \mathcal{ \bm G} = \{ \mathcal{G}_0, \mathcal{G}_1, ..., \mathcal{G}_T \}$, we calculate the snapshot anomaly score and consider the top ones as anomalous snapshots, adopting the same experiment settings in \cite{yoon2019fast}. 
We use both $\ell_1$ and $\ell_2$ as distance function to test the practicality of \textsc{DynAnom}.
We modify edge-level SedanSpot for Graph-level anomaly detection by aggregating edge anomaly score with  $\textsc{Mean}()$ \footnote{ For a fair comparison, we tried $\{ Mean(), Median(), Min(), Max() \}$ operators for SedanSpot, and found that $Mean()$ yields the best results.} in each snapshot. We describe detailed hyper-parameter settings in Appendix. \ref{sec:appendix-param-config}.
For DAPRA dataset, as suggested by \citet{yoon2019fast}, we take the precision at top-$250$ \footnote{ $k=250$ reflects its practicality since it is the closest $k$ without exceeding total 288 graph-level anomalies} as the performance measure in Table \ref{tab:top-k-precision}, and present Figure \ref{fig:darpa-rank-topk-precision} of recall-precision curves as we vary the parameter $k \in \{ 50,100,\ldots,800 \}$  for top-$k$ snapshots considered to be anomalous. 
% Enron emails records notorious Enron scandal has been intensively studied in \red{cite!}. 
% For the ENRON graph without anomaly ground-truth, the common evaluation is to visualize the anomalous scores against time, then interpret or associate the peaks with the real-world events, as listed in Table \ref{tab:enron-events} in Appendix.  

\paragraph{Practicality for Graph-level Anomaly}: 
The precision-recall curves show that \textsc{DynAnom} have a good trade-off, achieving the competitive second/third highest precision by predicting the highest top-$150$ as anomalous snapshots, and the precision decreases roughly at a linear rate when we includes more top-$k$  snapshots as anomalies. 
Table \ref{tab:top-k-precision} presents the precision at top-$250$, and we observe that \textsc{DynAnom} could outperform all other strong baselines, including its unweighted counterpart -- \textsc{DynPPE}. It further demonstrates the high practicality of the proposed flexible framework even with the simplest $\ell_1 \text{and }\ell_2$ distance function. 
 \vspace{-1mm}
\begin{table}[ht]
	\begin{minipage}{0.6\linewidth}
		\centering
		\includegraphics[width=1\textwidth]{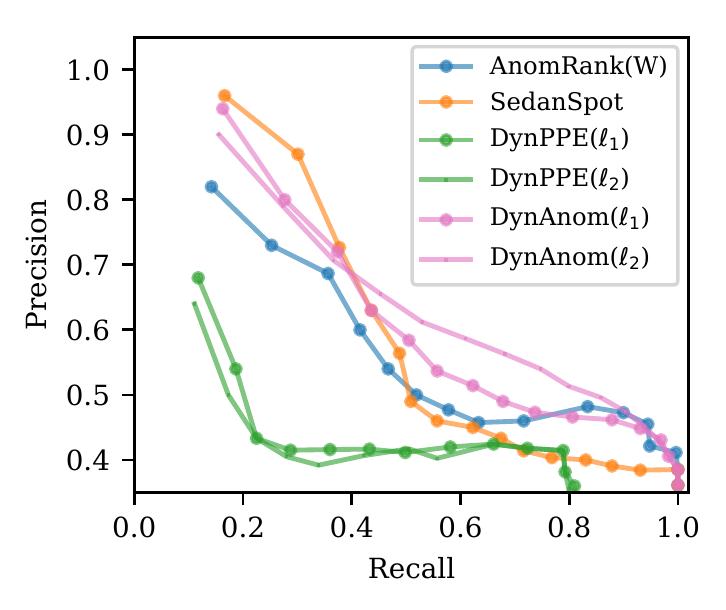}
		\vspace{-10mm}
		\captionof{figure}{Precision-Recall Curve}
		\label{fig:darpa-rank-topk-precision}
	\end{minipage}%
	\hfill
	\begin{minipage}{0.39\linewidth}
	\small
		\caption{The precision of top-250 anomalies}
		\vspace{-4mm}
		\label{tab:top-k-precision}
		\centering
        \begin{tabular}{lr}
            \toprule
                Algorithms &  Precision \\
            \midrule
                NetWalk &     OOM  \\
                AnomRankW &     0.5400  \\
                SedanSpot &     0.5640  \\
                DynPPE($\ell_1$) &     0.4160 \\
                DynPPE($\ell_2$) &     0.3920 \\
                DynAnom($\ell_1$) &     0.5840 \\
                \textbf{DynAnom($\ell_2$)} &     \textbf{0.6120} \\
            \bottomrule
        \end{tabular}
	\end{minipage}

\end{table}
\vspace{-1mm}

% \xingzhiswallow{
% \begin{minipage}{0.5\textwidth}
%     \hspace{-7mm}
%     \begin{minipage}[b]{0.60\textwidth}
%         \includegraphics[width=1\textwidth]{figs/fig-res-dataset_darpa-initSS_256-timeStep_60-topk-recall-precision-v2.pdf}
%     \vspace{-9mm}
%     \captionof{figure}{Precision-Recall Curve}
%     \label{fig:darpa-rank-topk-precision}
%     \end{minipage}
%     %
%     \begin{minipage}[b]{0.39\textwidth}
%     \small
%     %\vspace{-8mm}
%         \captionof{table}{The precision of top-250 anomalies.}
%         \vspace{-2mm}
%         \begin{tabular}{lr}
%             \toprule
%                 Algorithms &  Precision \\
%             \midrule
%                 NetWalk &     OOM  \\
%                 AnomRankW &     0.5400  \\
%                 SedanSpot &     0.5640  \\
%                 DynPPE($\ell_1$) &     0.4160 \\
%                 DynPPE($\ell_2$) &     0.3920 \\
%                 DynAnom($\ell_1$) &     0.5840 \\
%                 \textbf{DynAnom($\ell_2$)} &     \textbf{0.6120} \\
%             \bottomrule
%         \end{tabular}
%         \label{tab:top-k-precision}
%     \end{minipage}
% \end{minipage}

% }

% Enron emails records notorious Enron scandal has been intensively studied in \red{cite!}. 
% For the ENRON graph without anomaly ground-truth, the common evaluation is to visualize the anomalous scores against time, then interpret or associate the peaks with the real-world events, as listed in Table \ref{tab:enron-events} in Appendix.  

ENRON graph records the emails involved in the notorious Enron company scandal, and it has been intensively studied in \cite{eswaran2018sedanspot,chang2021f, yoon2019fast}.  Although ENRON graph has no explicit anomaly annotation, a common practice \cite{chang2021f,yoon2018fast} is to visualize the detected anomaly score against time and discover the associated real-world events. Likewise, Figure \ref{fig:enron-curve} plots the detected peaks side-by-side with other strong baselines \footnote{All values are post-processed by dividing the maximum value. For SedanSpot, we use $Mean()$ to aggregate edge anomaly score for the best visualization. Other aggregator (e.g., $Max()$) may produce flat curve, causing an unfair presentation. }. We find that \textsc{DynAnom} shares great overlaps with meaningful peaks and detects more prominent peak at Marker.1 when Enron got investigated for the first time,  demonstrating its effectiveness for sensible change discovery . 
%\footnote{All values are post-processed by dividing the maximum value. }
% In Figure \ref{fig:enron-curve}, We compare our method $DynAnom(\ell_2)$ to other strong baselines, and find that they share great overlaps with meaningful peaks, demonstrating the effusiveness of our method. 
% and annotate with the events for interpretation. 
For better interpretation, we annotate the peaks associated to the events collected from the Enron Timeline \footnote{Full events are included in Table \ref{tab:enron-events} in Appendix. Enron timeline: \url{https://www.agsm.edu.au/bobm/teaching/BE/Enron/timeline.html}}, and present some milestones as followings:
\begin{itemize}
\item[1]:(2000/08/23) Investigated Enron  when stock was high.
% \item[5]:(2001/05/17) Enron Chairman had a "Secret"  meeting with state governor. 
\item[6]:(2001/08/22) CEO was warned of accounting irregularities.
\item[10]:(2002/01/30) CEO was replaced after bankruptcy.
\end{itemize}
% \vspace{-5mm}

\vspace{-2mm}
\begin{figure}[htbp]
    \centering
    \includegraphics[width=0.47\textwidth]{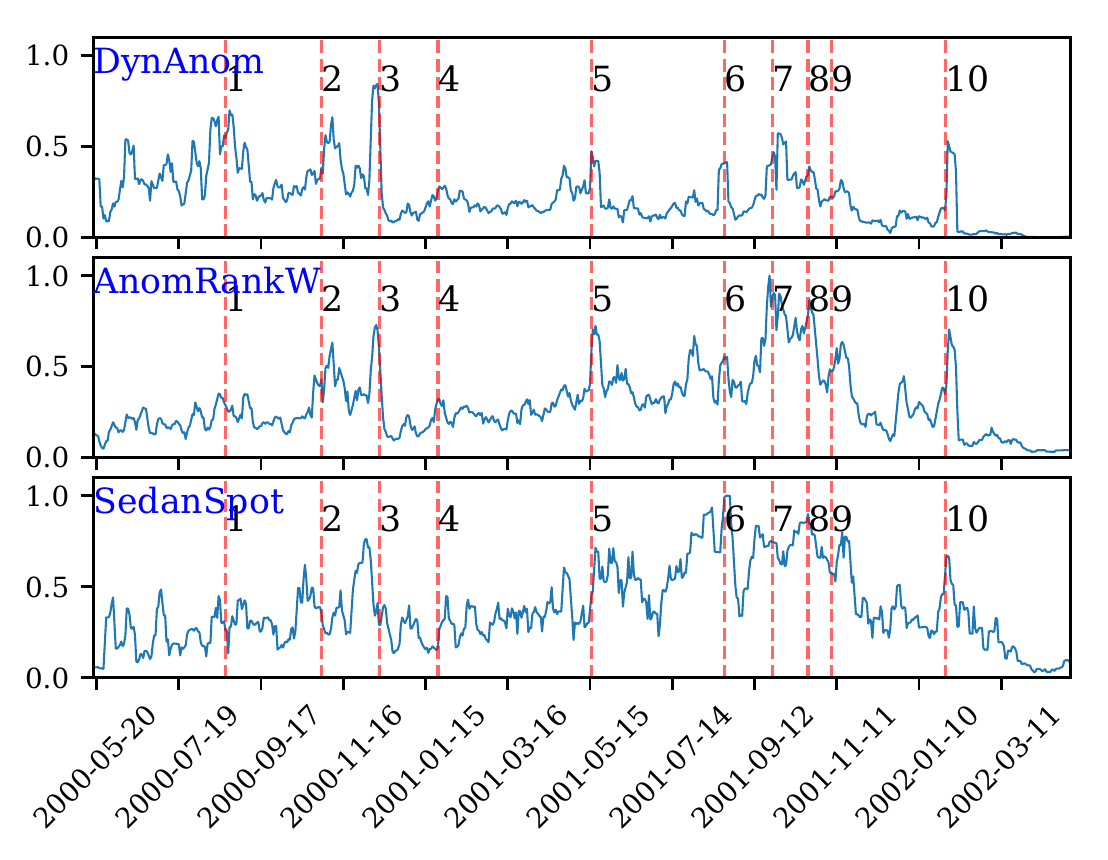}
    \vspace{-5mm}
    \caption{ The anomaly scores of ENRON graph. \textsc{DynAnom} well correlates with other baselines, and detect novel peaks at beginning.}
    \label{fig:enron-curve}
\end{figure}
% \vspace{-3mm}

% \paragraph{Observations}: We observe 

% We compare AnomRank, SadenSpot and our proposed score for global graph-level anomaly detection.

% \red{add different epsilon? or mean ?  discussion about findings }

% \textcolor{red}{Figure \ref{fig:darpa-score-timeline} and Figure \ref{fig:enron-score-timeline} } show the anomaly scores from baselines on DARPA and ENRON dataset.

% NetWalk is inefficient if we have finer-grained snapshots, it cannot finish in 10 hours, given ~1K snapshots over DARPA dataset.

% darpa precision-recall
\xingzhiswallow{

\begin{figure}[ht]
    \centering
    \includegraphics[width=0.35\textwidth]{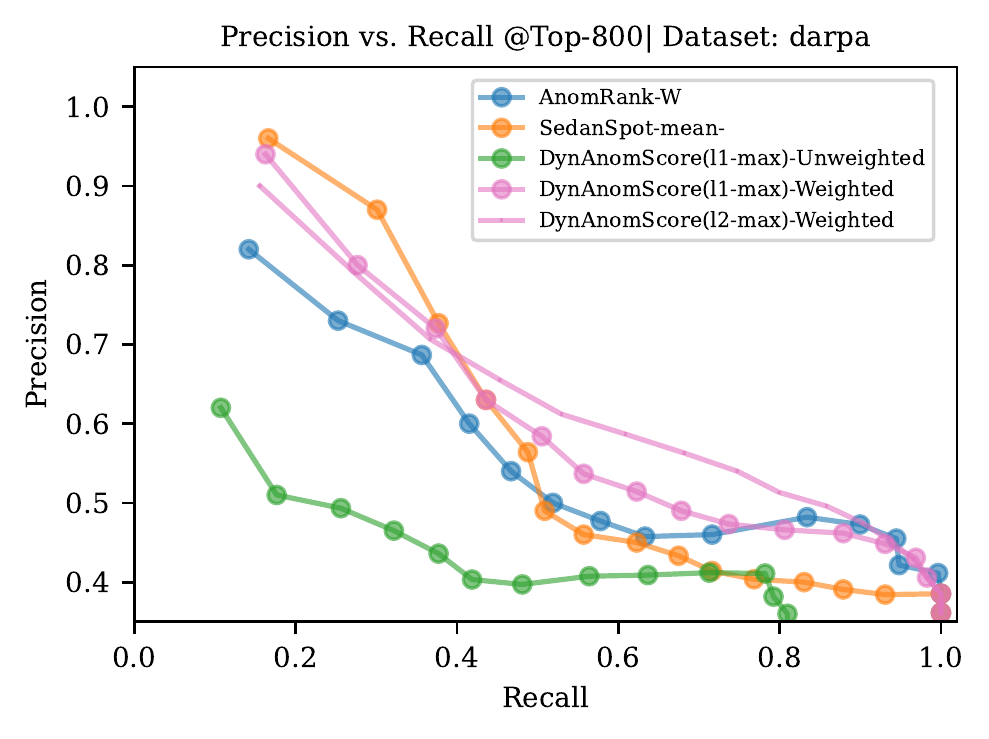}
    \caption{DARPA network: Precision-Recall for top-800}
    \label{fig:darpa-rank-topk-precision}
\end{figure}

\begin{table}[ht]
    \centering
    \small
\begin{tabular}{lr}
\toprule
          Anomaly score &  DARPA \\
\midrule
    NetWalk &     OOM  \\
    AnomRankW &     0.5400  \\
    SedanSpot &     0.5640  \\
    DynPPE &     0.4360 \\
    \textbf{DynAnom} &     \textbf{0.6120} \\

\bottomrule
\end{tabular}
    \caption{The precision of top-250 detected anomalous snapshots on DARPA dataset.}
\end{table}

}

\xingzhiswallow{

    \begin{figure*}[ht]
    \begin{subfigure}{0.5\textwidth}
       \centering
        \includegraphics[width=1.0\linewidth]{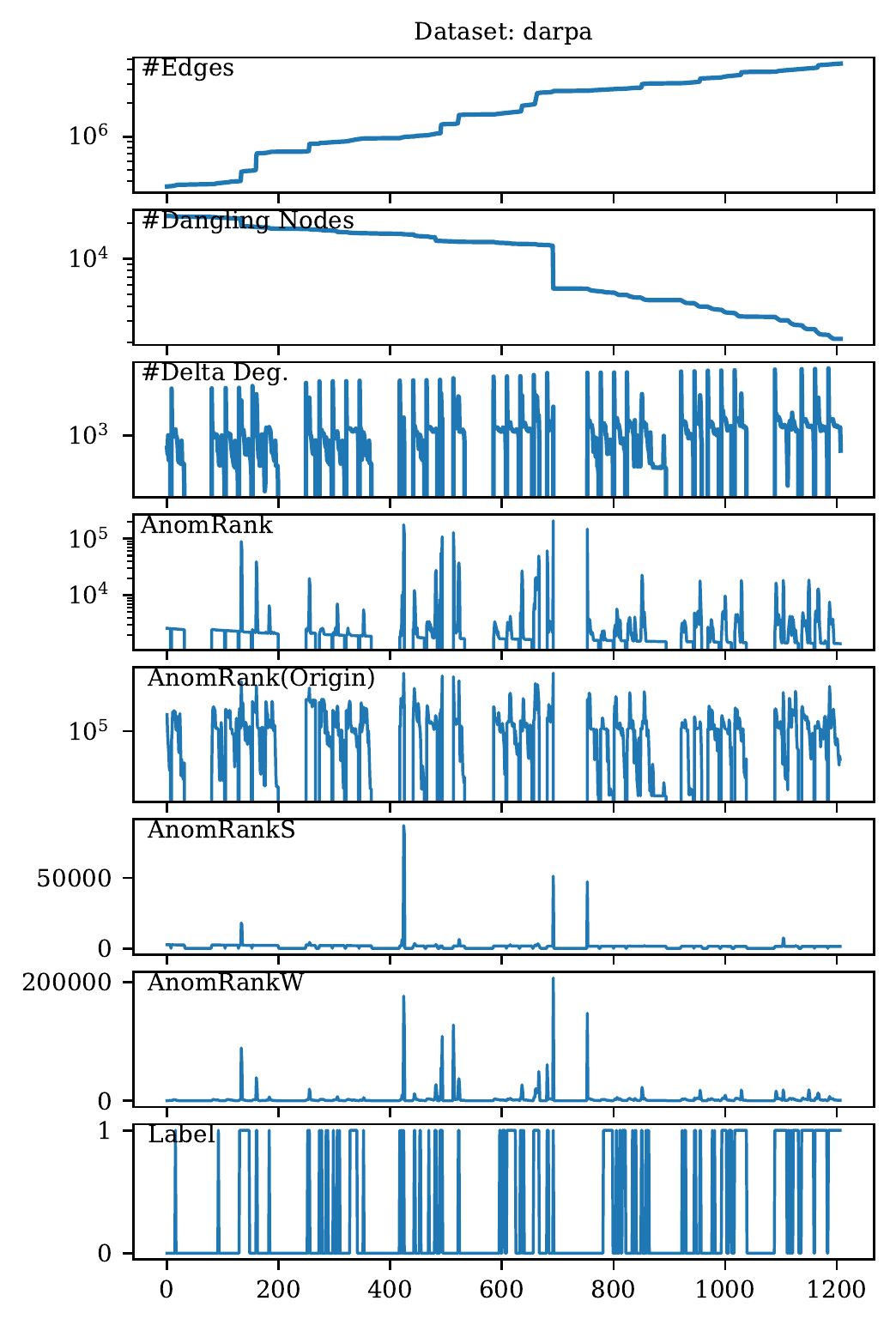}
        \caption{DARPA network: AnomRank Anomaly Scores vs. Ground-truth. }
        \label{fig:darpa-anomrank-scores}
    \end{subfigure}\hfill
    \begin{subfigure}{.5\textwidth}
        \centering
        \includegraphics[width=1.0\linewidth]{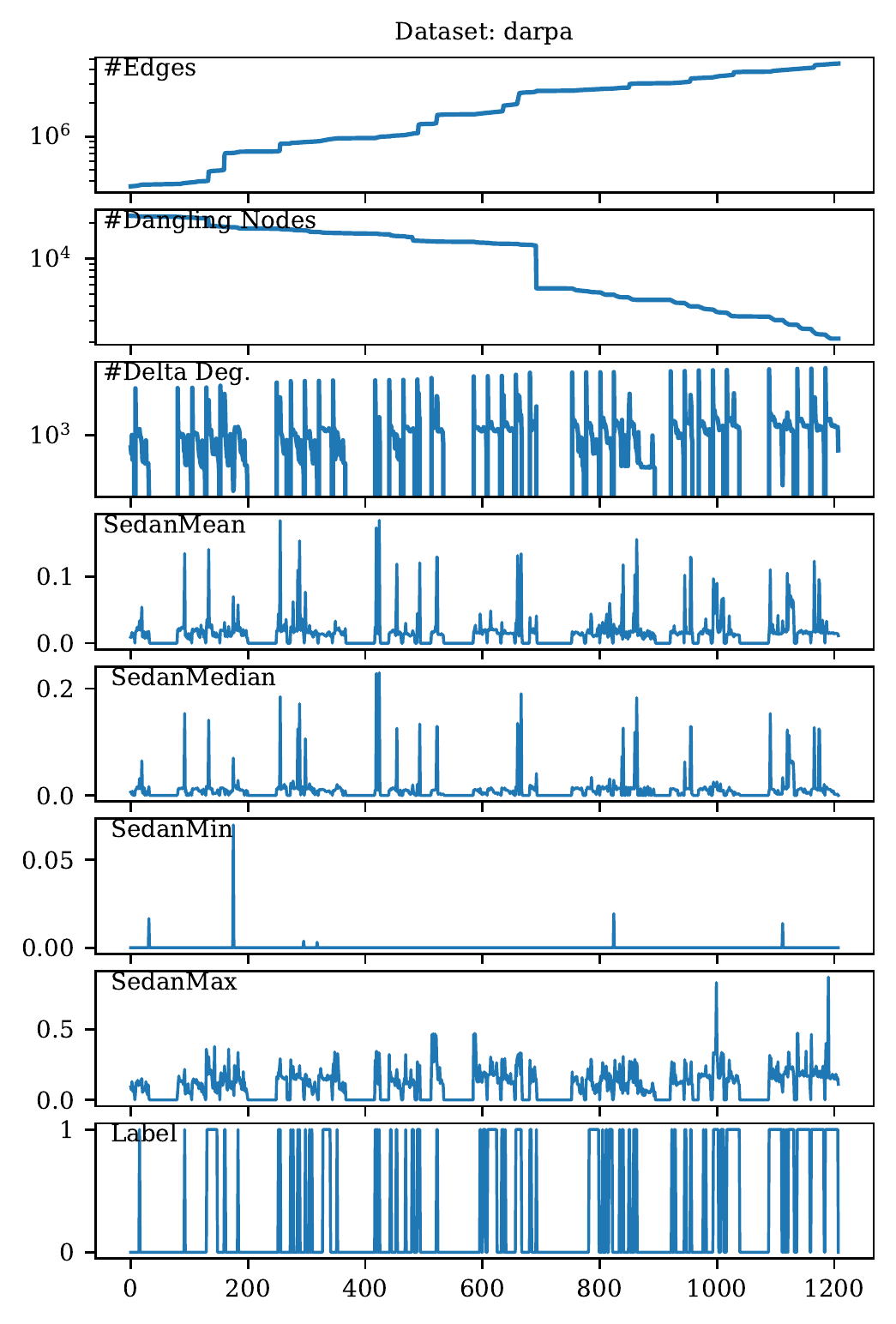}
        \caption{DARPA network: SedanSpot Anomaly Scores vs. Ground-truth}
        \label{fig:darpa-sedanSpot-scores}
    \end{subfigure}
    \caption{The baselines on DARPA dataset}
    \label{fig:darpa-score-timeline}
    \end{figure*}

\begin{figure*}[ht]
\begin{subfigure}{0.5\textwidth}
    \centering
    \includegraphics[width=0.7\linewidth]{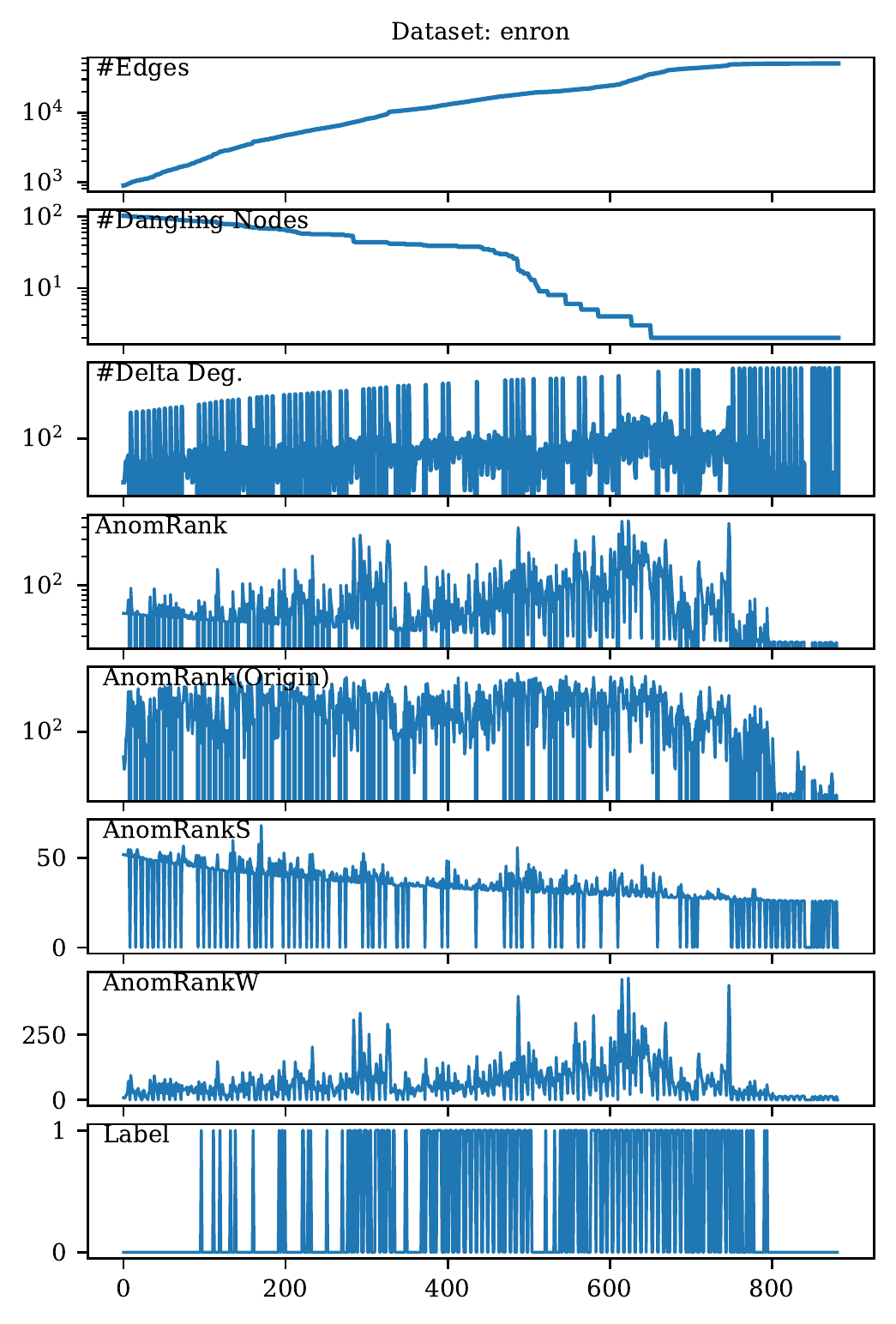}
    \caption{Enron network: AnomRank Anomaly Scores (No ground truth).}
    \label{fig:enron-anomrank-scores}
\end{subfigure}\hfill
\begin{subfigure}{.5\textwidth}
    \centering
    \includegraphics[width=0.7\linewidth]{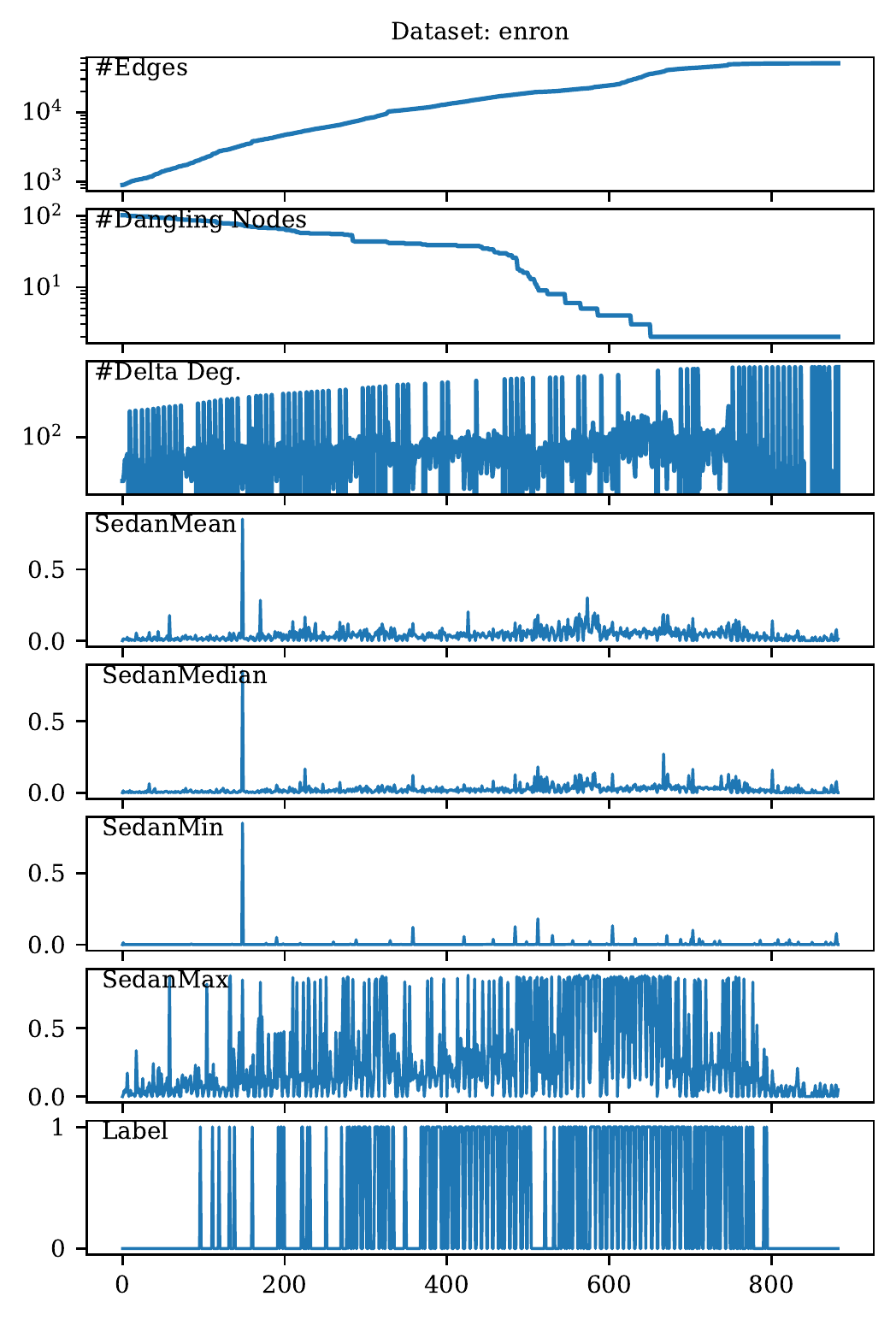}
    \caption{Enron network: SedanSpot Anomaly Scores (No ground truth).}
    \label{fig:enron-sedanSpot-scores}
\end{subfigure}
\begin{subfigure}{.5\textwidth}
    \centering
    \includegraphics[width=0.7\linewidth]{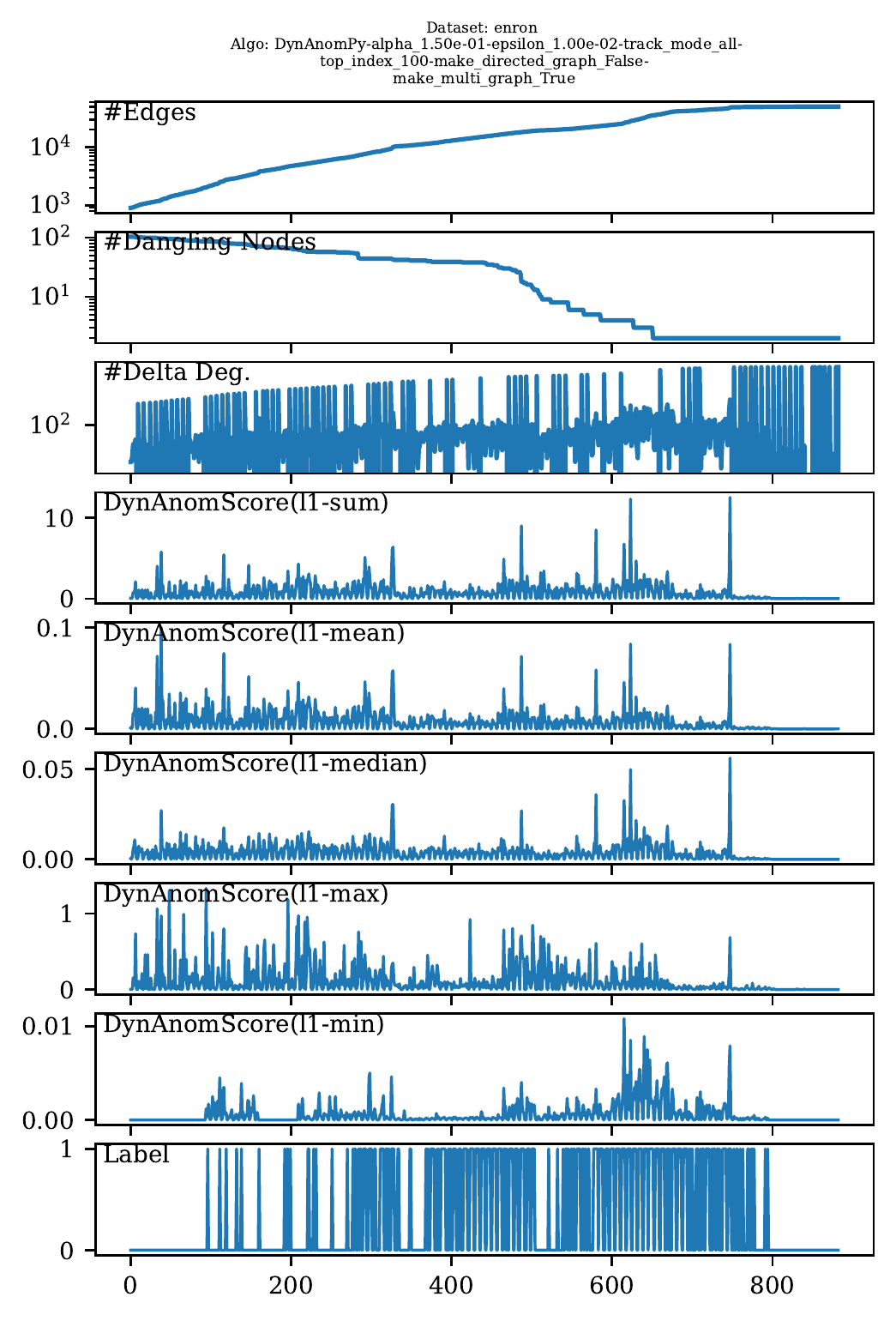}
    \caption{Enron network: \textcolor{red}{The Proposed} DynAnom Anomaly Scores (No ground truth).}
    \label{fig:enron-dynAnom-scores}
\end{subfigure}\hfill
\begin{subfigure}{.5\textwidth}
    \centering
    \includegraphics[width=1.3\linewidth]{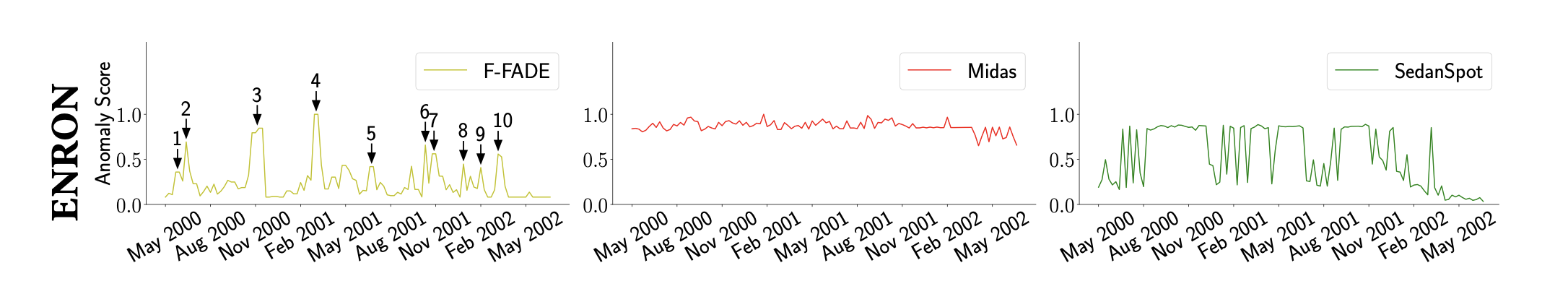}
    \caption{Enron network: From F-Fade, which is believed to be the best detection over Enron dataset.}
    \label{fig:enron-f-fade-scores}
\end{subfigure}

\caption{The baselines on Enron dataset. Our proposed method detects anomalies as peaks shown in Figure \ref{fig:enron-dynAnom-scores}}.
\label{fig:enron-score-timeline}
\end{figure*}

}

% darpa 
\xingzhiswallow{
\begin{figure}[ht]
    \centering
    \includegraphics[width=0.7\linewidth]{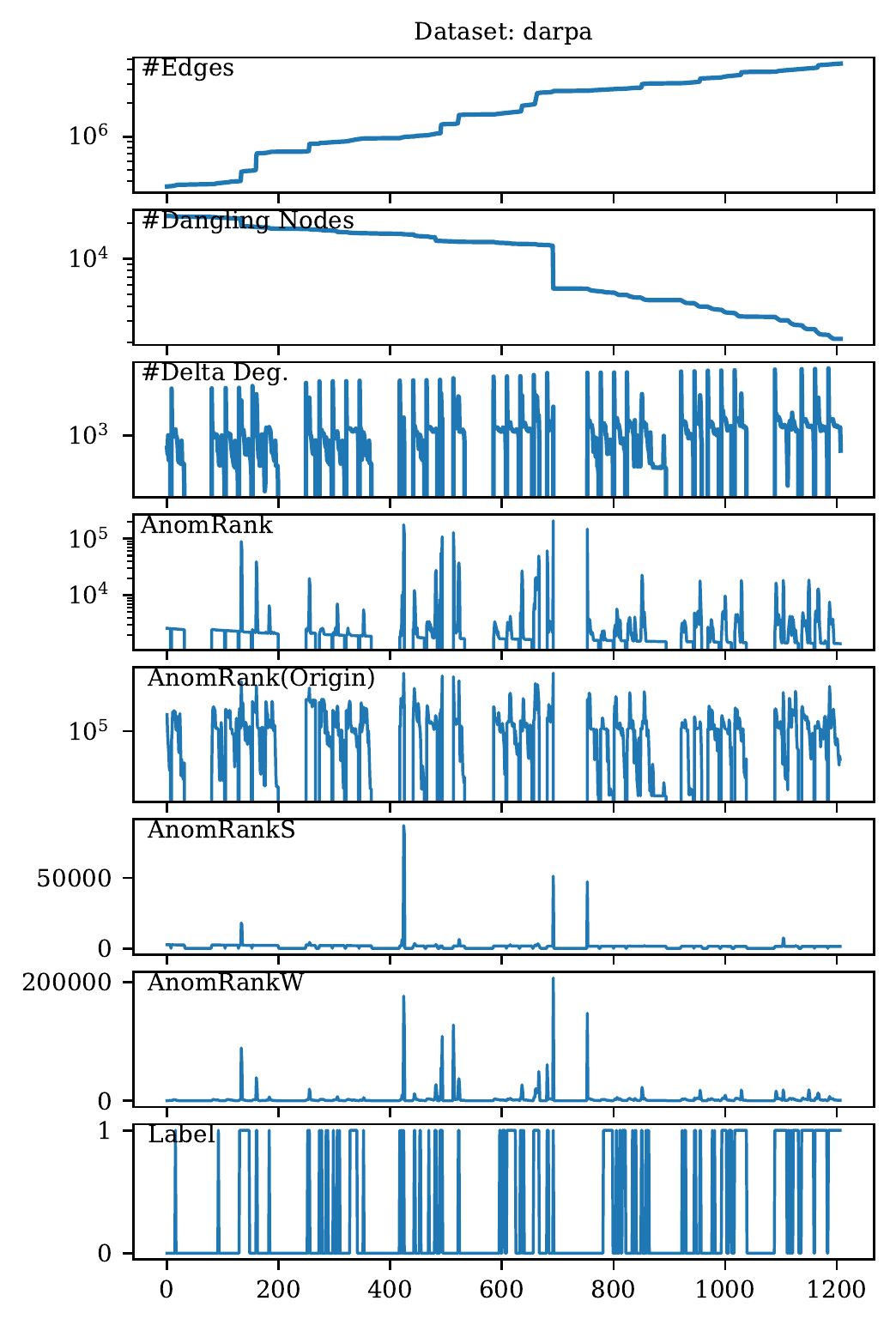}
    \caption{DARPA network: Anomaly Scores vs. Ground-truth. }
    \label{fig:darpa-anomrank-scores}
\end{figure}

\begin{figure}[ht]
    \centering
    \includegraphics[width=0.7\linewidth]{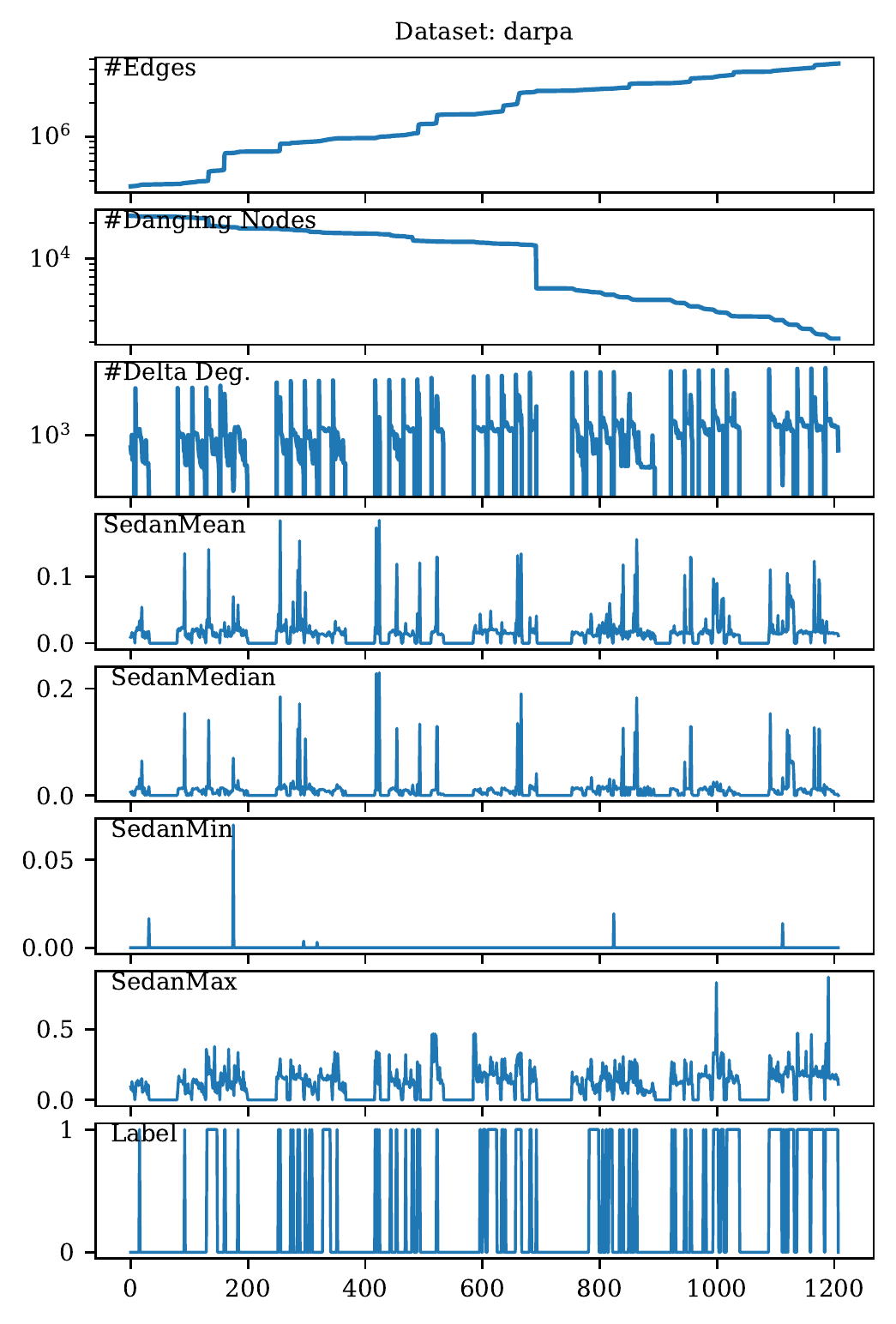}
    \caption{Enron network: SedanSpot Anomaly Scores (No ground truth).}
    \label{fig:darpa-sedanSpot-scores}
\end{figure}

}

% enron 
\xingzhiswallow{
\begin{figure}[ht]
    \centering
    \includegraphics[width=0.7\linewidth]{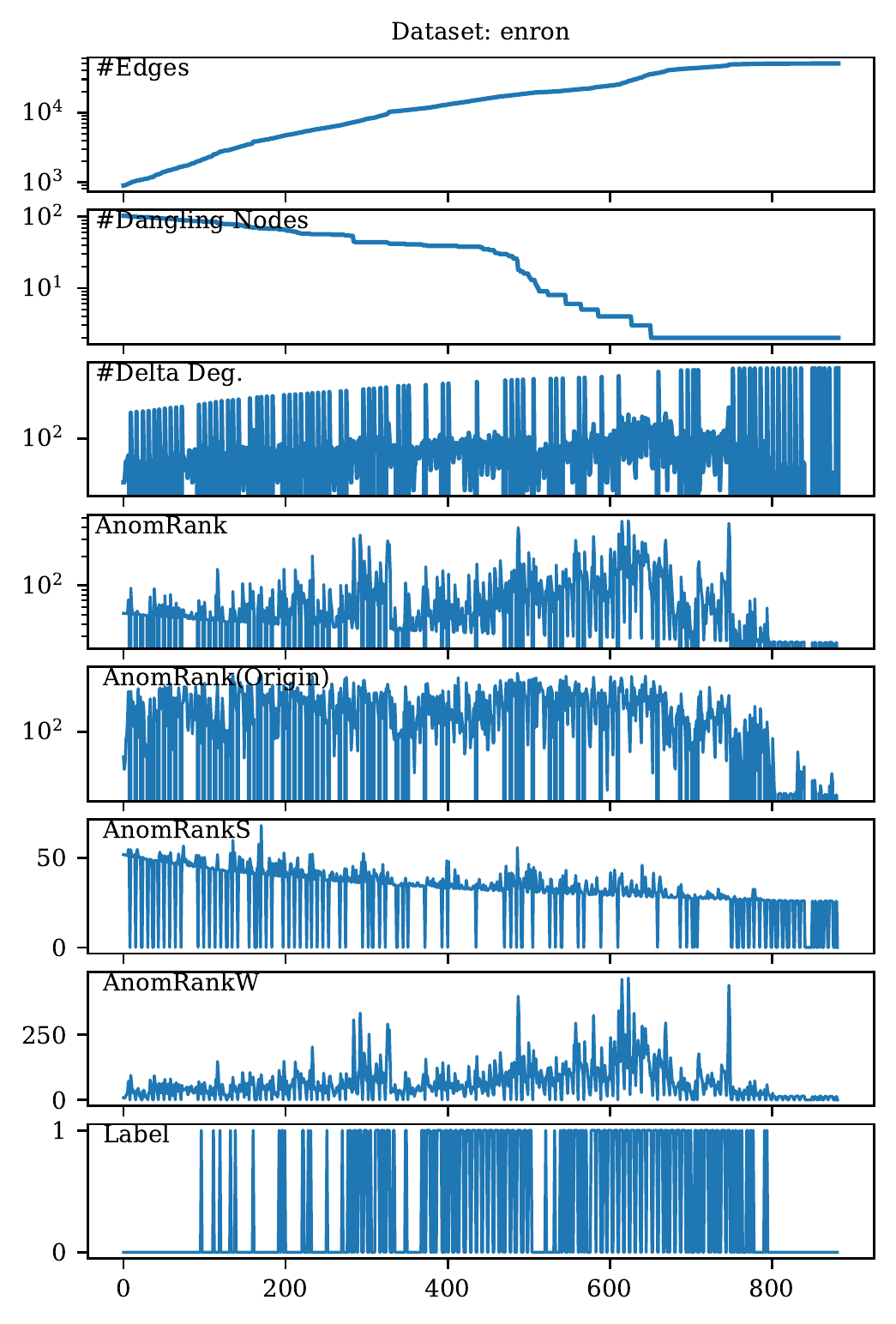}
    \caption{Enron network: Anomaly Scores (No ground truth).}
    \label{fig:enron-anomrank-scores}
\end{figure}

\begin{figure}[ht]
    \centering
    \includegraphics[width=0.7\linewidth]{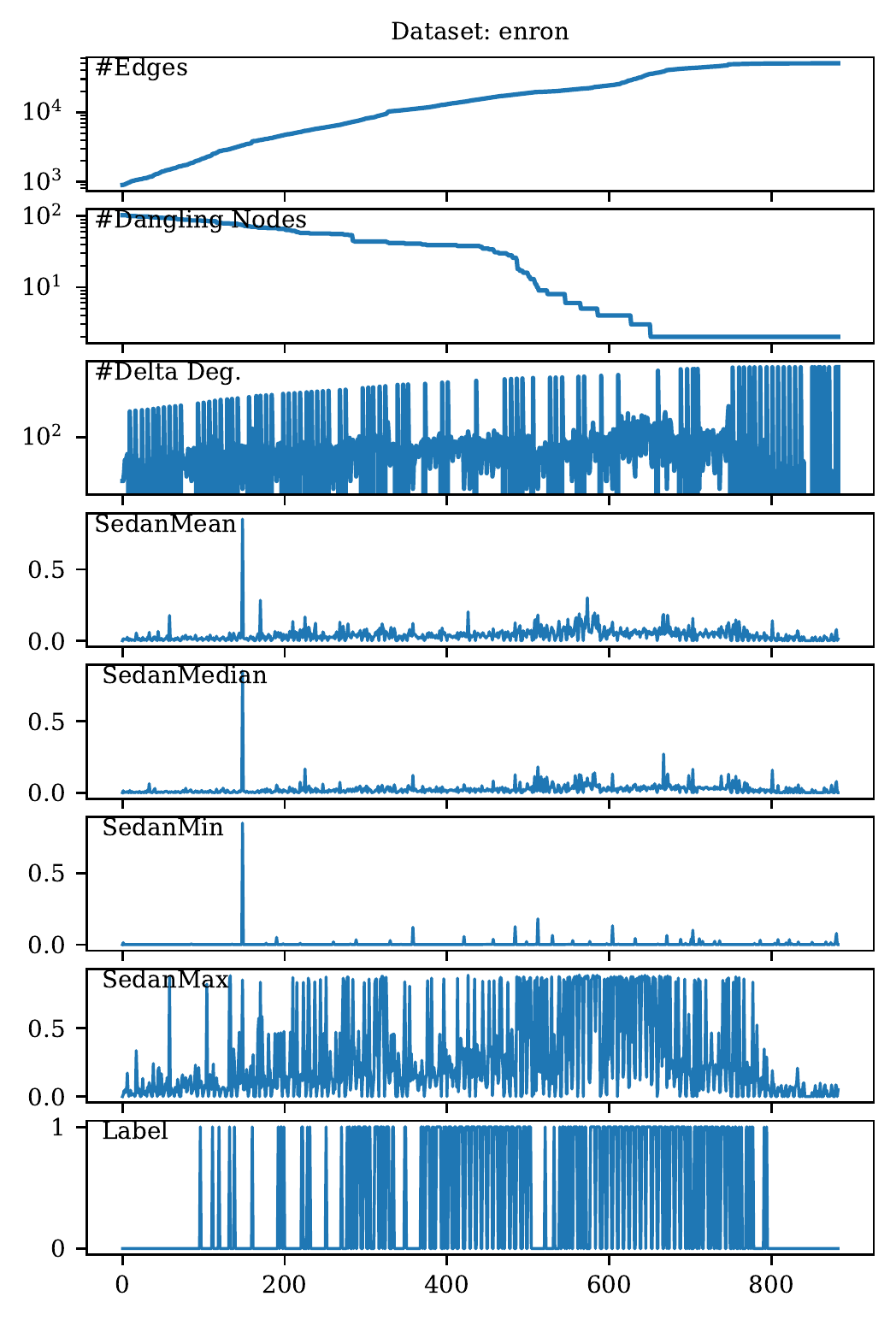}
    \caption{Enron network: SedanSpot Anomaly Scores (No ground truth).}
    \label{fig:enron-sedanSpot-scores}
\end{figure}

}

% \paragraph{ENRON Email Dataset}

% NetWalk\cite{yu2018netwalk} experiment setting:  Use first 50\% edges without any injected fake edge, then train clusters. Inject anomaly (fake) edges in the last 50\% and use embeddings to predict whether it's far from any trained cluster centers. For each testing snapshot, it use both previous fake and real edges to update the node embeddings. It is more similar to link prediction instead of anomaly detection.

\subsection{Case Study: Localize Changes in Person's life over Real World Graph}

% In our case studies, we want to detect the changes of individual entities in real-world graphs. Specifically 
To demonstrate the scalability and usability of \textsc{DynAnom},  we present an interesting case study over our constructed large-scale PERSON graph, which records the structure and intensity of person-event interaction history. 
We apply \textsc{DynAnom} to track public figures (Joe Biden, Arnold Schwarzenegger, Al Franken) of the U.S. from 2000 to 2022 on a yearly basis, and visualize their changes by $\ell_1$-distance, together with annotated peaks in Figure \ref{fig:eventkg-person}. 

\vspace{-3mm}
\begin{figure}[ht]
    \centering
    \includegraphics[width=1.0\linewidth]{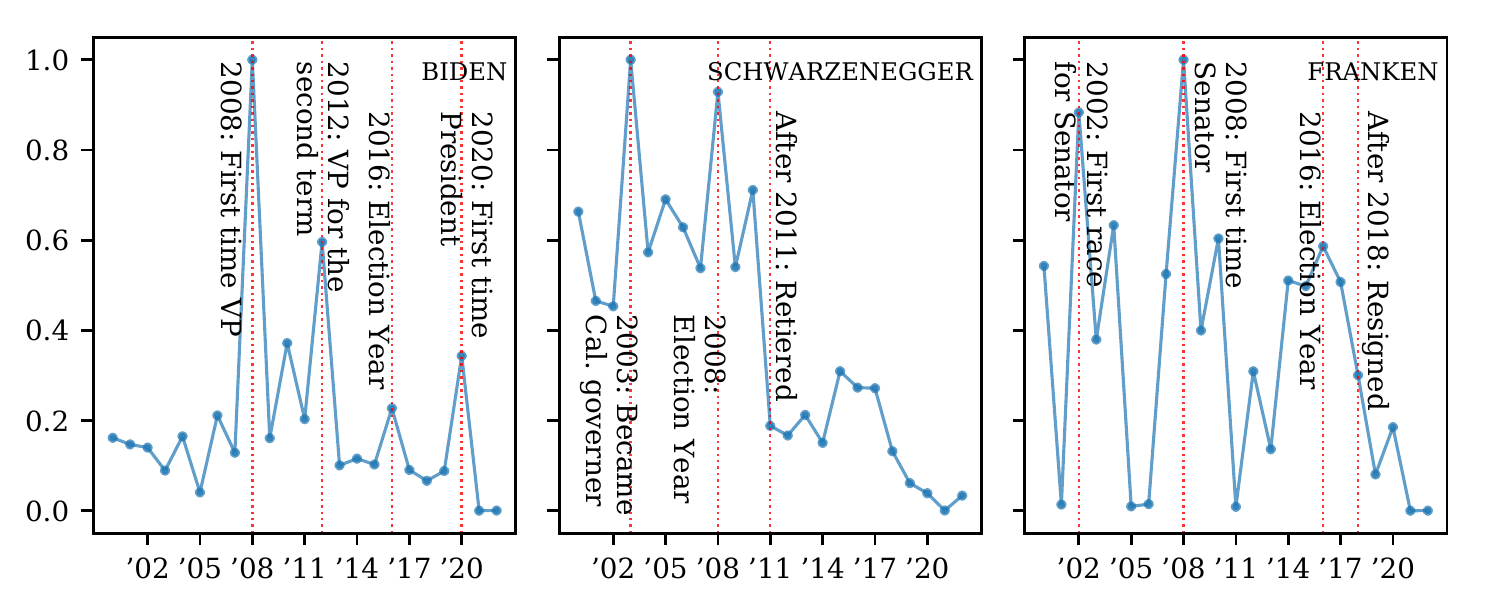}
    \vspace{-8mm}
    \caption{The yearly changes of public figures in PERSON graph from 2000-2022: The detected peaks are well correlated to the major events of the individuals as we annotated. }
    \label{fig:eventkg-person}
\end{figure}
% \vspace{-2mm}

% Figure \ref{fig:eventkg-person} shows three individuals with major career changes in 22 years from 2000 to 2022, where we plot the yearly difference \footnote{All values are post-processed by dividing the maximum value. } and annotate the peaks for each individual.

% Specifically, we track public figures (e.g., Joe Biden, Arnold Schwarzenegger, Al Franken) of the U.S. from 2000 to 2021 in a year basis. 

\xingzhiswallow{
\textit{Biden} has his greatest peak in year 2008 when he became the 47-th U.S. Vice President in Obama's administration for the very first time. The second peak occurs in 2012 when he won the re-election. This time the peak has smaller magnitude because he was already VP, the re-election caused slightly less difference in him when compared to his transition in 2008. The third peak is in 2020 when he won the U.S. Presidential election.
% In year 2020, Biden has a great peak in degree change, but the status may node change so much since he had been already Vice President since 2008, where the major peak occurred. 
\textit{Arnold Schwarzenegger} won the California gubernatorial election as his first politician role in 2013. This has been captured by the highest peak in node status change although the degree change is not the greatest. In the year 2008, Arnold began campaigning with McCain as the key endorsement for McCain's presidential campaign.  After 2011, Arnold reached his term limit as Governor and returned to acting.
\textit{Al Franken} is a famous comedian who shfited career to be a politician, the peaks well capture the time when he entered political area and went through elections in each year. Table \ref{tab:person-events} lists full events in appendix.
}

As we can see, \textit{Biden} has his greatest peak in year 2008 when he became the 47-th U.S. Vice President for the very first time, which is considered to be his first major advancement. The second peak occurs in 2012 when he won the re-election. This time the peak has smaller magnitude because he was already the VP, thus the re-election caused less difference in him, compared to his transition in 2008. The third peak is in 2020 when he became the president, the  magnitude is even smaller as he had been the VP for 8 years without huge context changed. %, although having great political impact.
% In year 2020, Biden has a great peak in degree change, but the status may node change so much since he had been already Vice President since 2008, where the major peak occurred. 
Similarly, the middle sub-figure captures  \textit{Arnold Schwarzenegger}'s transition to be California Governor in 2003, and high activeness in election years. 
%
%shows \textit{Arnold Schwarzenegger} who transited from an actor to a California Governor in 2003 with the highest peak. Followed by election year 2008 when he joined campaigning as the key endorsement.  %After 2011, Arnold reached his term limit as Governor and returned to acting.
%
\textit{Al Franken} is also a famous comedian who shifted career to be a politician in 2008, the peaks well capture the time when he entered political area and the election years. Table \ref{tab:person-events} lists full events in Appendix \ref{sec:appendix-timelines}. This resources could bring more opportunities for knowledge discovery and anomaly mining studies in the graph mining community.

Moreover, based on our current Python implementation, it took roughly 37.8 minutes on a 4-core CPU machine to track the subset of three nodes over the PERSON graph, which has more than 600k nodes and 8.7 million edges.  Both presented results demonstrate that our proposed method has great efficiency and usability for practical subset node tracking.

% where we observed peaks related to the individual when s/he had major career events.
% demonstrating that the detected peaks are closely correlated to the well-known real life events. 

% Specifically, we track two major countries (China and the U.S.) in \textit{COUNTRY} graph from 1980 to 2022, where we observed several peaks associated to the country involving in major conflicts. 
% Meanwhile, we also track public figures (e.g., Joe Biden, Arnold Schwarzenegger, Al Franken) of the U.S. in the \textit{PERSON} graph from 2000 to 2021, where we observed peaks related to the individual when s/he has major career events.

% \paragraph{PERSON Graph}
%2245.452 (edge adjust) + 10.516 (push )

% \vspace{-2mm}

%  17565258 2229.887    0.000 2245.452    0.000 DynamicPPE.py:287(dynamic_adjust)
%       43  218.194    5.074 2661.072   61.885 DynamicPPE.py:445(get_approx_ppv_with_new_edges_batch_multip_source)
%  17565258  114.060    0.000  114.060    0.000 DynamicPPE.py:416(update_graph)
% 105394548   33.739    0.000   33.739    0.000 typedlist.py:280(_numba_type_)
%         1   25.702   25.702   46.287   46.287 graph_data_utils.py:12(read_data)
%  52697651   15.566    0.000   15.566    0.000 serialize.py:29(_numba_unpickle)
%  26957907   14.991    0.000   26.550    0.000 std.py:1159(__iter__)
%         1   14.639   14.639 2737.509 2737.509 test_dynmaicPPE_exp3_case_studies.py:109(test_graph_stream)
%       43   10.516    0.245   10.516    0.245 DynamicPPE.py:151(push_numba_parallel)

% \paragraph{Wiki-Graph Dataset}

\xingzhiswallow{
\textbf{COUNTRY Graph} records the conflict between countries and entities. We investigated China and the U.S. and found interesting peaks associated to global/regional tension as presented in Figure \ref{fig:gdelt-china-us} and full event list in Table \ref{tab:country-events} in appendix. 

\vspace{-2mm}
\begin{figure}[ht]
    \centering
    \includegraphics[width=1.0\linewidth]{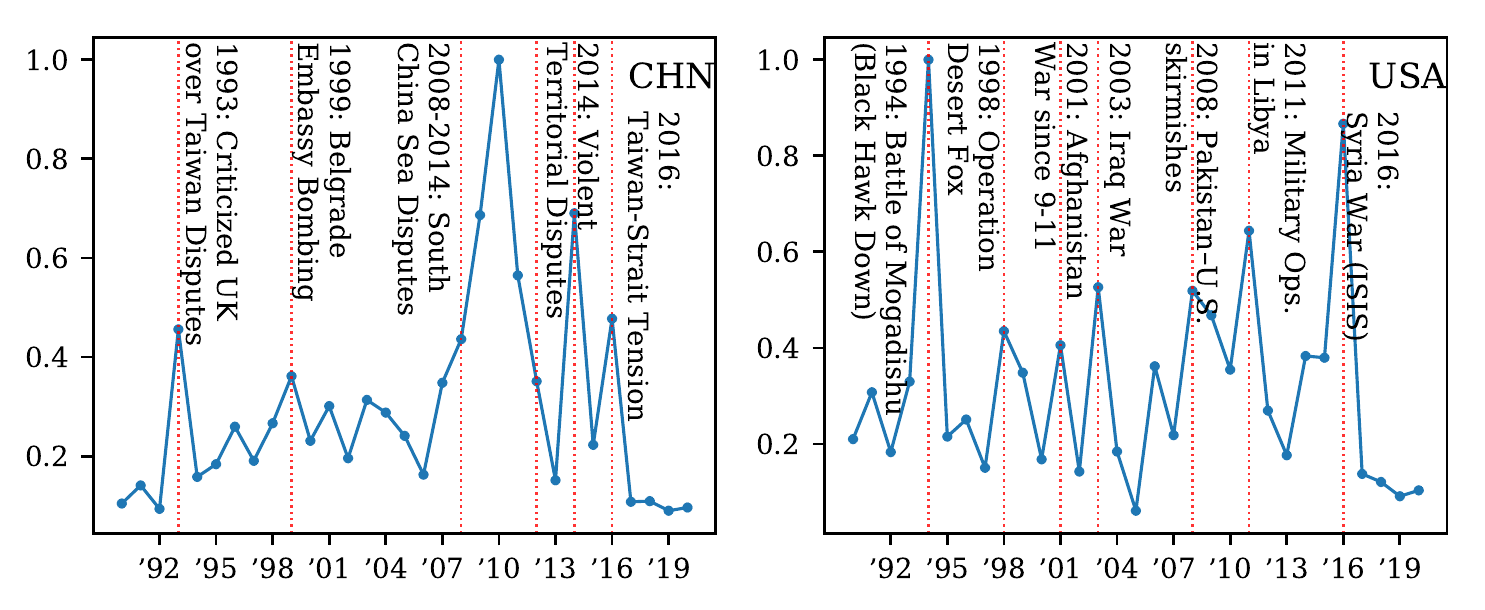}
    \vspace{-8mm}
    \caption{ The yearly changes of two countries in  \textit{COUNTRY} graph from 1990-2020 : The detected peaks reflect the major conflicts as the countries shift their diplomatic/militant behaviors.  }
    \label{fig:gdelt-china-us}
\end{figure}
}

% \begin{figure}[ht]
%     \centering
%     \includegraphics[width=1.0\linewidth]{figs/fig-anold-temporal-gamma-1.pdf}
%     \caption{Arnold Schwarzenegger: 1979-2021}
%     \label{fig:eventkg-arnold}
% \end{figure}

\section{Discussion and Conclusion}
\label{sec:conclusion}

In this paper, we propose an unified framework \textsc{DynAnom} for subset node anomaly tracking over large dynamic graphs. This framework can be easily applied to different graph anomaly detection tasks from local to global with a flexible score function customized for various application. 
Experiments show that our proposed framework outperforms current state-of-the-art methods by a large margin, and has a significant speed advantage (2.3 times faster) over large graphs.  
We also present a real-world PERSON graph with an interesting case study about personal life changes, providing a rich resource for both knowledge discovery and algorithm benchmarking.
For the future work, it remains interesting to explore different type of score functions, automatically identify the interesting subset of nodes as better strategies for tracking global-level anomaly, and further investigate temporal-aware PageRank as better node representations.

%\textcolor{red}{future work, score function, S for graph, temporal pagerank.}

\begin{acks}
This work was partially supported by NSF grants IIS-1926781, IIS-1927227, IIS-1546113 and OAC-1919752.
\end{acks}

\bibliographystyle{ACM-Reference-Format}
\balance
\bibliography{references}

\clearpage
\appendix
 
\section{Proof of Thm. \ref{theorem:dyn-adjust-weighted-graph}}
\label{appendix-proof}
\begin{proof}
By the invariant property in Lemma \ref{lemma:ppr-invar}, we have
\begin{align}
p_s(i) + \alpha r_s(i) = (1-\alpha)\sum_{x \in N^{in}(i)}  \frac{ w_{(x,i)} p_s(x)}{d(x)} + \alpha \times 1_{i=s}, \forall i \in \mathcal{V},  \label{eq:u-invaraint-appendix} 
\end{align}

Initially, this invariant holds and keeps $ \frac{r_s(i)}{d(i)} \leq \epsilon$ after applying Algo.\ref{algo:forward-local-push-weighted}. 
When the new edge $(u,v,\Delta w_{(u,v)})$ arrives, it changes the out-degree of $u$, breaking up the balance for all invariant involving $d(u)$. Our goal is to recover such invariant by adjusting small amount of $p_s, r_s$. While such adjustments may compromise the quality of PPVs, we could incrementally invoke Algo.\ref{algo:forward-local-push-weighted} to update $p_s$ and $r_s$ for better accuracy afterwards. We denote the initial vectors as $p_s, r_s, d$ and post-change vectors as $p'_s, r'_s, d'$. 

As $d'(u)$ is involved in $p_s(u) / d'(u)$, one need to have $p'_s(u)$ such that $ p_s(u) / d(u) = p'_s(u) / d'(u)$, i.e. the invariant maintenance of Equ. \eqref{eq:u-invaraint-appendix}. The updated amount of weight is $\delta w_(u,v)$, which indicates $p'_s(u) =( d(u)+\Delta w_{(u,v)} ) p_s(u) / d(u)$. So, we have
\begin{align}
p'_s(u) = p_s(u) \frac{\sum_{v \in \operatorname{Nei}(u)} w_{(u,v)} + \Delta w_{(u,v)}}{\sum_{v \in \operatorname{Nei}(u)} w_{(u,v)}}. \label{eq:p-update-weighted-appendix}
\end{align}
Hence, we reach
Equ. \ref{eq:p-update-weighted-appendix} implies that we should scale up $p_s(u)$ to recover balance. This strategy is the general case of \citet{zhang2016approximate} for unweighted graphs where $\Delta w_{(u,v)} = 1$.

However, once we assign $p'_s(u)$, it breaks the invariant for $u$ since $p_s(u)+\alpha r_s(u) \neq p'_s(u)+\alpha r_s(u)$. Likewise, we maintain this equality by introducing $r'_s(u)$:
\begin{align}
p_s(u) + \alpha r_s(u) & = p_s'(u) + \alpha r_s'(u) \nonumber\\
&\text{Substitute $p'_s(u)$ from eq.\ref{eq:p-update-weighted-body}} \nonumber \allowdisplaybreaks\\
& = p_s(u) + \frac{\Delta w_{(u,v)} p_s(u)}{d(u)} + \alpha r_s'(u) , \nonumber \allowdisplaybreaks \\
\Longrightarrow \alpha r_s'(u) - \alpha r_s(u) &= -\frac{\Delta w_{(u,v)} p_s(u)}{d(u)},\nonumber \allowdisplaybreaks\\
r_s'(u) &=  r_s(u) -\frac{\Delta w_{(u,v)} p_s(u)}{\alpha d(u)} 
\label{eq:r-update-weighted-appendix}
\end{align}

Equ. \eqref{eq:r-update-weighted-appendix} implies that we should decrease $r_s(u)$. From a heuristic perspective, it's consistent to the behavior of BCA-algorithm which keeps taking mass from $r_s$ to $p_s$. In this updating case, we artificially create mass for $p'_s$ at the expense of $r_s$'s decrements. When $\Delta w_{(u,v)} = 1$ in case of unweighted graphs, this update rule is also equivalent to \citet{zhang2016approximate}.

So far, we have the updated $p'_s(u)$ and $r'_s(u)$ so that all nodes (without direct edge update, except for $v$) keep the invariant hold. However, due to the introduction of $\Delta w_{(u,v)}$, the shares of mass pushed to $v$ is changes, breaking the invariant for $v$ as shown below:
\begin{align}
    p_s(v) + \alpha r_s(v) & \neq (1-\alpha) \Bigl( \frac{ (w_{(u,v)} + \Delta w_{(u,v)})  p'_s(u)}{d'(u)} \Bigr. \nonumber\\
    & \Bigl. + \sum_{x \in N^{in}(v) \setminus \{ u \} }  \frac{w_{(u,x)} p_s(x)}{d(x)} \Bigr) +\alpha \times 1_{t=s} \nonumber
\end{align}

In order to recover the balance with minimal effort, we should update $r_s(v)$ instead of $p_s(v)$. The main reason is that any change in $p_s(v)$ will break the balance for $v$'s neighbors, similar to the breaks incurred by the change of $p_s(u)$. We present the updated $r'_s(v) = r_s(v)+\Delta $ as following:
\begin{align}
p_s(v) + \alpha \Bigl(r_s(v)+\Delta \Bigr) = (1-\alpha) \Bigl( \frac{\Delta w_{(u,v)} p'_s(u)}{d'(u)} + \Bigr. \nonumber\\
\Bigl. \sum_{x \in N^{in}(v) }  \frac{w_{(u,x)} p_s(x)}{d(x)} \Bigr) +\alpha 1_{t=s}  \label{eq:invar-v-body}.    
\end{align}
Note that $p'_s(u) / d'(u) = p_s(u) / d(u)$, and
\[
p_s(v) + \alpha r_s(v) = (1-\alpha)\sum_{x \in N^{in}(v)}  \frac{ w_{(x,v)} p_s(x)}{d(x)}.
\]
Reorganize Equ. \eqref{eq:invar-v-body} and cancel out $p_s(v), r_s(v)$, we have
\begin{align}
\Delta &= \frac{(1-\alpha)}{\alpha} \frac{ \Delta_{w_{(u,v)}} p_s(u)}{d(u)} \nonumber \allowdisplaybreaks\\
r_s'(v)  &= r_s(v) + \frac{(1-\alpha)}{\alpha} \frac{\Delta_{w_{(u,v)}}p_s(u)}{d(u)} \label{eq:r-v-update-weighted-body-appendix}
\end{align}
Combining Equ.  \eqref{eq:p-update-weighted-appendix},\eqref{eq:r-update-weighted-appendix}, and \eqref{eq:r-v-update-weighted-body-appendix}, we prove the theorem.
\end{proof}
\begin{remark}
The above proof mainly follows from \cite{zhang2016approximate} where unweighted graph is considered while we consider weighted graph in our problem setting.
\end{remark}

\section{Proof of Thm. \ref{thm:time-complexity}}

Before we proof the theorem, we present the known time complexity of \textsc{IncrementPush} as in the following lemma.
\begin{lemma}[Time complexity of \textsc{IncrementPush} \cite{xingzhi2021subset}]
Suppose the teleport parameter of obtaining PPR is $\alpha$ and the precision parameter is $\epsilon$. Given current weighted graph snapshot $\mathcal{G}_t$ and a set of edge events $\Delta E_t$ with $|\Delta E_t| = m$ , the time complexity of \textsc{IncrementPush} is $\mathcal{O}(m/\alpha^2 + \bar{d}^t / (\epsilon \alpha^2) +  1/(\epsilon \alpha))$ where $\bar{d}^t$ is the average node degree of current graph $\mathcal{G}_t$.
\label{lemma:run-time-increment-push}
\end{lemma}
\begin{proof}
The total time complexity of our instantiation \textsc{DynAnom} has components, which correspond to three steps in Algo. \ref{algo:dynanom-node}. The run time of step 1, \textsc{IncrementPush} is $\mathcal{O}(m/\alpha^2 + \bar{d}^t / (\epsilon \alpha^2) + 1/(\epsilon \alpha))$ by Lemma \ref{lemma:run-time-increment-push}. The run time of step 2, \textsc{DynNodeRep} is bounded by $\mathcal{O}(n)$ as the component of \textsc{DynNodeRep} is the procedure of the dimension reduction by using two hash functions. The calculation of two has functions given $\bm p_s$ is linear on $|\operatorname{supp}(\bm p_s)|$, which is $|\operatorname{supp}(\bm p_s)| \leq n$. Finally, the instantiation of our score function is also linear on $n$. Therefore, the whole time complexity is dominated by $\mathcal{O}(k m/\alpha^2 + k \bar{d^t} / (\epsilon \alpha^2) +  k /(\epsilon \alpha) + k T s)$ where $s$ is the maximal allowed sparsity defined in the theorem. We proof the theorem.
\end{proof}
 
\section{Timelines of real world graphs}
\label{sec:appendix-timelines}
We list the real-world events for the ENRON and PERSON graph in table \ref{tab:enron-events} and \ref{tab:person-events}.

\begin{table}[h!]
    \centering
    \caption{The events in Enron scandal timeline}

    \small
\begin{tabular}{p{0.05\linewidth} p{0.15\linewidth} p{0.65\linewidth} }
\toprule
 Index &       Date &  Event Description \\
\midrule
     %1 & 2000/08/23 &           Stock hits all-time high of \$90.56. Market valuation of \$70 billion. FERC (the Federal Energy Regulatory Commission) orders an investigation into strategies designed to drive electricity prices up in California. \\
    1 & 2000/08/23 &           Stock hits all-time high of \$90.56. the Federal Energy Regulatory Commission orders an investigation. \\

     2 & 2000/11/01 &                                                                                                                                                         FERC investigation exonerates Enron for any wrongdoing in California. \\
     3 & 2000/12/13 &                                                               Enron announces that president and chief operating officer Jeffrey Skilling will take over as chief executive in February. Kenneth Lay will remain as chairman. \\
     4 & 2001/01/25 &                                                                   Analyst Conference in Houston, Texas. Skilling bullish on the company. Analysts are all convinced. \\
     5 & 2001/05/17 &                                                                                                                                                     "Secret" meeting at Peninsula Hotel in LA -- Schwarzenegger, Lay, Milken. \\
     6 & 2001/08/22 &                                                                                                                   Ms Watkins meets with Lay and gives him a letter in which she says that Enron might be an "elaborate hoax." \\
     7 & 2001/09/26 &                                                                                                              Employee Meeting. Lay tells employees: Enron stock is an "incredible bargain." "Third quarter is looking great." \\
     8 & 2001/10/22 &                                                                   Enron acknowledges Securities and Exchange Commission inquiry into a possible conflict of interest related to the company's dealings with the partnerships. \\
     9 & 2001/11/08 & Enron files documents with SEC revising its financial statements to account for \$586 million in losses. The company starts negotiations to sell itself to head off bankrutcy. \\
    10 & 2002/01/30 &                                                                                                                                                                 Stephen Cooper takes over as Enron CEO. \\
\bottomrule
\end{tabular}
    \label{tab:enron-events}
\end{table}

\begin{table}[h!]
    \centering
    \small
    \caption{The person events of \textit{Schwarzenegger} and \textit{Franken}}
\begin{tabular}{p{0.1\linewidth} p{0.77\linewidth}}
\toprule
 Year & Major Event \\
\midrule 
 \multicolumn{2}{c}{Arnold Schwarzenegger} \\
\midrule 
 2003 & Arnold Schwarzenegger won the California gubernatorial election as his first politician role\\
 2008 & Arnold began campaigning with McCain as the key endorsement for McCain's presidential campaign\\
 2010 & mid-term election\\
 2011 & Arnold reached his term limit as Governor and returned to acting\\
\midrule 
 \multicolumn{2}{c}{Al Franken} \\
\midrule 
 2002 & Al Franken consider his first race for office due to the tragedy of Minnesota Sen. Paul Wellstone. \\
  2004 & The Al Franken Show aired \\
 2007 & Al Franken announced for candidacy for Senate. \\
 2008 & Al Franken won election. \\
 2012 & Al Franken won re-election. \\
 2018 & Al Franken resigned \\
 \bottomrule
\end{tabular}
    \label{tab:person-events}
\end{table}

\section{Hyper-parameter Settings}
we (re)implement the algorithms in Python to measure comparable running time.
\label{sec:appendix-param-config}
We list the hyper-parameter in Table \ref{tab:param-config}.
\begin{table}[h]
    \small
    \centering
\caption{The hyper-parameter configurations of algorithms in experiments. }
\begin{tabular}{p{0.20\linewidth}p{0.65\linewidth}}
\toprule 
 Algorithm & Hyper-parameter Settings \\
\midrule 
 AnomRank & alpha = 0.5 (default ) , \ epsilon = 1e-2 (default ) or 1e-6 (for better accuracy)  \\
\hline 
 SedanSpot & sample-size = 500 (default settings) \\
\hline 
 NetWalk & epoch\_per\_update=10, dim=128, default settings for the rest parameters \\
\hline 
 DynAnom \& DynPPE & For exp1 and exp2: alpha = 0.15, minimum epsilon = 1e-12 , dim=1024; For exp2: we keep track of the top-100 high degree nodes for graph-level anomaly.  For case studies: alpha = 0.7, the rest are the same. \\
 \bottomrule
\end{tabular}
\label{tab:param-config}
\end{table}

\section{Experiment Details}
\label{sec:appendix-data}

\paragraph{Infrastructure}: We conduct our experiment on machine with 4-core Intel i5-6500 3.20GHz CPU, 32 GB memory, and GeForce GTX 1070 GPU (8 GB memory) on Ubuntu 18.04.6 LTS.

\paragraph{Datasets:}
For DARPA dataset: we totally track 200 anomalous nodes, and 151 of them have at least one anomalous edge after initial snapshot, the ground-truth of node-level anomaly is derived from the annotated edge as aforementioned. 
For EU-CORE dataset: Since EU-CORE dataset does not have annotated anomalous edges,  we randomly select 20 snapshots to inject artificial anomalous edges using two popular injection methods,  which are similar to the approach used in \cite{yoon2019fast}. For each selected snapshot in \textit{EU-CORE-S}, we uniformly select one node,$u_{high}$,  from the top 1\% high degree nodes, and injected 70 multi-edges connecting the selected node to other 10 random nodes which were not connected to $u_{high}$ before, which simulates the structural changes.  In each selected snapshot in \textit{EU-CORE-L}, we uniformly select 5 pairs of nodes, and injected edge connect each pair with totally 70 multi-edges as anomaly, which simulates the sudden peer-to-peer communication.  We include all datasets in our supplementary materials.
% The raw dataset DARPA and ENRON are accessible at \url{https://github.com/minjiyoon/KDD19-AnomRank/blob/master/darpa.txt} and \url{https://networkrepository.com/ia-enron-employees.php}.

\end{document}

%% file: fig-joe-biden.tex
\begin{figure}[t]
\centering
    \subfloat[Dynamic Network] {
        \begin{tikzpicture}[x=0.75pt,y=0.75pt,yscale=-0.65,xscale=0.60]
        %uncomment if require: \path (0,270); %set diagram left start at 0, and has height of 270
        
        %Rounded Rect [id:dp40437996572522916] 
        \draw   (15,27.2) .. controls (15,15.49) and (24.49,6) .. (36.2,6) -- (167.8,6) .. controls (179.51,6) and (189,15.49) .. (189,27.2) -- (189,90.8) .. controls (189,102.51) and (179.51,112) .. (167.8,112) -- (36.2,112) .. controls (24.49,112) and (15,102.51) .. (15,90.8) -- cycle ;
        %Shape: Ellipse [id:dp2646329502205623] 
        \draw   (49.41,32.59) .. controls (49.41,27.3) and (54.49,23.02) .. (60.74,23.02) .. controls (67,23.02) and (72.07,27.3) .. (72.07,32.59) .. controls (72.07,37.88) and (67,42.16) .. (60.74,42.16) .. controls (54.49,42.16) and (49.41,37.88) .. (49.41,32.59) -- cycle ;
        %Shape: Ellipse [id:dp7428451069893581] 
        \draw  [fill={rgb, 255:red, 0; green, 0; blue, 255 }  ,fill opacity=1 ] (84.66,57.41) .. controls (84.66,52.12) and (89.74,47.84) .. (96,47.84) .. controls (102.25,47.84) and (107.33,52.12) .. (107.33,57.41) .. controls (107.33,62.7) and (102.25,66.98) .. (96,66.98) .. controls (89.74,66.98) and (84.66,62.7) .. (84.66,57.41) -- cycle ;
        %Shape: Ellipse [id:dp18836392620264886] 
        \draw   (136,33.82) .. controls (136,28.53) and (141.08,24.25) .. (147.33,24.25) .. controls (153.59,24.25) and (158.67,28.53) .. (158.67,33.82) .. controls (158.67,39.11) and (153.59,43.39) .. (147.33,43.39) .. controls (141.08,43.39) and (136,39.11) .. (136,33.82) -- cycle ;
        %Shape: Ellipse [id:dp44787195692399373] 
        \draw   (128.75,96.84) .. controls (128.75,91.55) and (133.82,87.26) .. (140.08,87.26) .. controls (146.34,87.26) and (151.41,91.55) .. (151.41,96.84) .. controls (151.41,102.12) and (146.34,106.41) .. (140.08,106.41) .. controls (133.82,106.41) and (128.75,102.12) .. (128.75,96.84) -- cycle ;
        %Shape: Ellipse [id:dp8172237305157681] 
        \draw   (43.96,93.22) .. controls (43.96,87.93) and (49.03,83.65) .. (55.29,83.65) .. controls (61.55,83.65) and (66.62,87.93) .. (66.62,93.22) .. controls (66.62,98.51) and (61.55,102.79) .. (55.29,102.79) .. controls (49.03,102.79) and (43.96,98.51) .. (43.96,93.22) -- cycle ;
        %Straight Lines [id:da7774860552266111] 
        \draw    (60.74,42.16) -- (55.29,83.65) ;
        %Straight Lines [id:da029859917521621204] 
        \draw    (63.53,87) -- (88.39,64) ;
        %Straight Lines [id:da7533427358418158] 
        \draw    (107.33,57.41) -- (135.5,36.39) ;
        %Straight Lines [id:da06407721853912296] 
        \draw    (72.07,32.59) -- (136,33.82) ;
        %Straight Lines [id:da7273513737814138] 
        \draw [fill={rgb, 255:red, 100; green, 100; blue, 100 }  ,fill opacity=1 ] [dash pattern={on 0.84pt off 2.51pt}]  (50.51,27) -- (24.47,10) ;
        %Straight Lines [id:da25395336086823206] 
        \draw [fill={rgb, 255:red, 100; green, 100; blue, 100 }  ,fill opacity=1 ] [dash pattern={on 0.84pt off 2.51pt}]  (51.69,38) -- (16.18,49) ;
        %Straight Lines [id:da8462491471677348] 
        \draw [fill={rgb, 255:red, 100; green, 100; blue, 100 }  ,fill opacity=1 ] [dash pattern={on 0.84pt off 2.51pt}]  (43.96,93.22) -- (15,90.8) ;
        %Straight Lines [id:da49624875504666466] 
        \draw [fill={rgb, 255:red, 100; green, 100; blue, 100 }  ,fill opacity=1 ] [dash pattern={on 0.84pt off 2.51pt}]  (56.43,112) -- (55.29,100.79) ;
        %Straight Lines [id:da2068543421394059] 
        \draw [fill={rgb, 255:red, 100; green, 100; blue, 100 }  ,fill opacity=1 ] [dash pattern={on 0.84pt off 2.51pt}]  (174.8,109) -- (151.41,96.84) ;
        %Straight Lines [id:da8213255990163468] 
        \draw [fill={rgb, 255:red, 100; green, 100; blue, 100 }  ,fill opacity=1 ] [dash pattern={on 0.84pt off 2.51pt}]  (153.49,7) -- (147.33,24.25) ;
        %Straight Lines [id:da9555794669364805] 
        \draw [fill={rgb, 255:red, 100; green, 100; blue, 100 }  ,fill opacity=1 ] [dash pattern={on 0.84pt off 2.51pt}]  (185.45,47) -- (157.04,40) ;
        %Straight Lines [id:da7408053432401054] 
        \draw    (147.33,43.39) -- (140.08,87.26) ;
        %Rounded Rect [id:dp1751852590819647] 
        \draw   (15,167.2) .. controls (15,155.49) and (24.49,146) .. (36.2,146) -- (167.8,146) .. controls (179.51,146) and (189,155.49) .. (189,167.2) -- (189,230.8) .. controls (189,242.51) and (179.51,252) .. (167.8,252) -- (36.2,252) .. controls (24.49,252) and (15,242.51) .. (15,230.8) -- cycle ;
        %Shape: Ellipse [id:dp3042691730331203] 
        \draw   (50.13,172.59) .. controls (50.13,167.3) and (55.31,163.02) .. (61.7,163.02) .. controls (68.08,163.02) and (73.26,167.3) .. (73.26,172.59) .. controls (73.26,177.88) and (68.08,182.16) .. (61.7,182.16) .. controls (55.31,182.16) and (50.13,177.88) .. (50.13,172.59) -- cycle ;
        %Shape: Ellipse [id:dp7177157173574817] 
        \draw  [fill={rgb, 255:red, 0; green, 0; blue, 255 }  ,fill opacity=1 ] (86.12,197.41) .. controls (86.12,192.12) and (91.29,187.84) .. (97.68,187.84) .. controls (104.07,187.84) and (109.25,192.12) .. (109.25,197.41) .. controls (109.25,202.7) and (104.07,206.98) .. (97.68,206.98) .. controls (91.29,206.98) and (86.12,202.7) .. (86.12,197.41) -- cycle ;
        %Shape: Ellipse [id:dp6617481785890633] 
        \draw   (138.52,173.82) .. controls (138.52,168.53) and (143.7,164.25) .. (150.09,164.25) .. controls (156.48,164.25) and (161.66,168.53) .. (161.66,173.82) .. controls (161.66,179.11) and (156.48,183.39) .. (150.09,183.39) .. controls (143.7,183.39) and (138.52,179.11) .. (138.52,173.82) -- cycle ;
        %Shape: Ellipse [id:dp8768372644245978] 
        \draw   (131.12,236.84) .. controls (131.12,231.55) and (136.3,227.26) .. (142.69,227.26) .. controls (149.08,227.26) and (154.25,231.55) .. (154.25,236.84) .. controls (154.25,242.12) and (149.08,246.41) .. (142.69,246.41) .. controls (136.3,246.41) and (131.12,242.12) .. (131.12,236.84) -- cycle ;
        %Shape: Ellipse [id:dp9011893255045064] 
        \draw   (44.56,233.22) .. controls (44.56,227.93) and (49.74,223.65) .. (56.13,223.65) .. controls (62.52,223.65) and (67.69,227.93) .. (67.69,233.22) .. controls (67.69,238.51) and (62.52,242.79) .. (56.13,242.79) .. controls (49.74,242.79) and (44.56,238.51) .. (44.56,233.22) -- cycle ;
        %Straight Lines [id:da6196040367703768] 
        \draw    (61.7,182.16) -- (56.13,223.65) ;
        %Straight Lines [id:da33831908137175126] 
        \draw    (64.54,227) -- (89.92,204) ;
        %Straight Lines [id:da761365624150331] 
        \draw [color={rgb, 255:red, 144; green, 19; blue, 254 }  ,draw opacity=1 ][fill={rgb, 255:red, 144; green, 19; blue, 254 }  ,fill opacity=1 ][line width=3]    (109.25,197.41) -- (140.67,180) ;
        %Straight Lines [id:da6546773853993254] 
        \draw    (73.26,172.59) -- (138.52,173.82) ;
        %Straight Lines [id:da6618011229363271] 
        \draw  [dash pattern={on 0.84pt off 2.51pt}]  (51.25,167) -- (24.67,150) ;
        %Straight Lines [id:da08305733605818688] 
        \draw  [dash pattern={on 0.84pt off 2.51pt}]  (52.46,178) -- (16.21,189) ;
        %Straight Lines [id:da7212189005480917] 
        \draw  [dash pattern={on 0.84pt off 2.51pt}]  (44.56,233.22) -- (15,230.8) ;
        %Straight Lines [id:da008440198894443363] 
        \draw  [dash pattern={on 0.84pt off 2.51pt}]  (57.29,252) -- (56.13,240.79) ;
        %Straight Lines [id:da16770559556698839] 
        \draw  [dash pattern={on 0.84pt off 2.51pt}]  (178.13,249) -- (154.25,236.84) ;
        %Straight Lines [id:da3828076021315745] 
        \draw  [dash pattern={on 0.84pt off 2.51pt}]  (156.38,147) -- (150.09,164.25) ;
        %Straight Lines [id:da5223613448627548] 
        \draw  [dash pattern={on 0.84pt off 2.51pt}]  (189,187) -- (160,180) ;
        %Straight Lines [id:da5632994136521178] 
        \draw    (150.09,183.39) -- (142.69,227.26) ;
        %Straight Lines [id:da9923953282093889] 
        \draw [color={rgb, 255:red, 144; green, 19; blue, 254 }  ,draw opacity=1 ][fill={rgb, 255:red, 0; green, 0; blue, 0 }  ,fill opacity=1 ]   (105.63,205) -- (132.21,231) ;
        %Straight Lines [id:da6862349118069943] 
        \draw [line width=1.5]    (34.32,112) -- (34.32,143) ;
        % \draw [shift={(34.32,146)}, rotate = 270] [color={rgb, 255:red, 0; green, 0; blue, 0 }  ][line width=1.0]    (14.21,-4.28) .. controls (9.04,-1.82) and (4.3,-0.39) .. (0,0) .. controls (4.3,0.39) and (9.04,1.82) .. (14.21,4.28)   ;
        \draw [shift={(34.32,146)}, rotate = 270] [fill={rgb, 255:red, 0; green, 0; blue, 0 }  ][line width=0.5]  [draw opacity=0] (8.93,-4.29) -- (0,0) -- (8.93,4.29) -- cycle    ;

        % Text Node
        \draw (88.58,48) node [anchor=north west][inner sep=0.75pt]  [color={rgb, 255:red, 255; green, 255; blue, 255 }  ,opacity=1 ] [align=left] {u};
        % Text Node
        \draw (140.66,25) node [anchor=north west][inner sep=0.75pt]   [align=left] {v};
        % Text Node
        \draw (90.27,188) node [anchor=north west][inner sep=0.75pt]  [color={rgb, 255:red, 255; green, 255; blue, 255 }  ,opacity=1 ] [align=left] {u};
        % Text Node
        \draw (143.44,165) node [anchor=north west][inner sep=0.75pt]   [align=left] {v};
        % Text Node
        \draw (92.59,5.4) node [anchor=north west][inner sep=0.75pt]    {$\mathcal{G}_{t}$};
        % Text Node
        \draw (94.67,144.4) node [anchor=north west][inner sep=0.75pt]    {$\mathcal{G}_{t+1}$};
        % Text Node
        \draw (133.50,88) node [anchor=north west][inner sep=0.75pt]   [align=left] {k};
        % Text Node
        \draw (135.00,226) node [anchor=north west][inner sep=0.75pt]   [align=left] {k};
        % Text Node
        \draw (45,120.4) node [anchor=north west][inner sep=0.75pt]    {\tiny $ \Delta E =\{( u,v,1) ,( u,k,1)\}$};
        \end{tikzpicture}
    \label{fig:dyn-network-example}
    }
    \hfill
    \subfloat[Application: Quantify changes of Biden] {\includegraphics[width=0.62\linewidth]{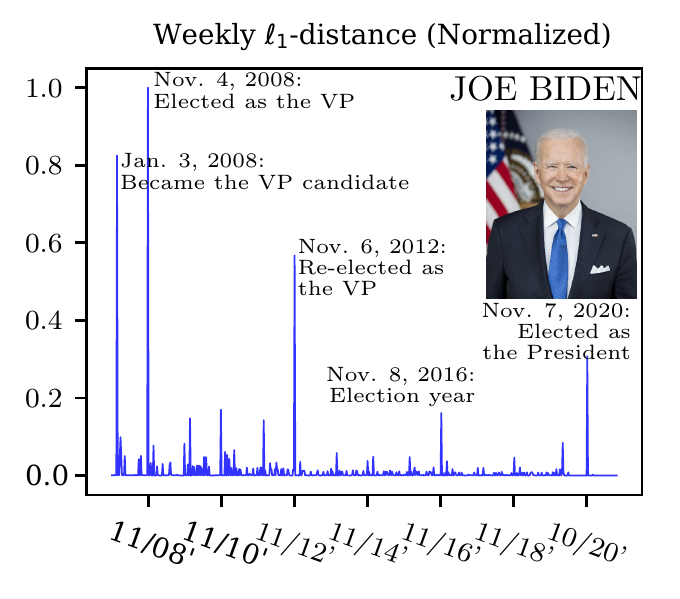}
     \label{fig:biden-life-example}
     }
     \vspace{-3mm}

\caption{ 
\footnotesize
(a) Two consecutive snapshots $\mathcal{G}_t$ and $\mathcal{G}_{t+1}$ from a dynamic weighted-graph. The inserted edge $(u, v, 1)$ further strengthens an existing relation, while edge $(u, k, 1)$ builds a new link.
(b) Anomaly tracking of Joe Biden over the \textit{Person Knowledge-Graph} using our proposed \textsc{DynAnom}.
We track Joe Biden on a weekly basis from 2007 to 2020, and calculate the $\ell_1$-distance between the representations on consecutive weeks. The detected peaks correlate well with significant changes in Biden's life.
}
\label{fig:biden-changes}
\vspace{-4mm}

\end{figure}

%% file: forward-push-algo.tex
\begin{algorithm}[ht]
\caption{$\textsc{DynamicForwardPush}$ \cite{zhang2016approximate}}
\begin{algorithmic}[1]
\State \textbf{Input: }$\bm p_s, \bm r_s, \mathcal{G}_t, \epsilon, \alpha$
\While{ exists $u$ such that $| r_s(u)| > \epsilon d(u)$}
\State $\textsc{Push}(u)$
\EndWhile
\State \Return $(\bm p_s, \bm r_s)$
\Procedure{Push}{$u$}
\State $p_s(u) \pluseq \alpha r_s(u)$
\For{$v \in \operatorname{Nei}(u)$}
\State $r_s(v) \pluseq \frac{(1-\alpha)  r_s(u)w_{(u, v)}}{ \sum_{v \in \operatorname{Nei}(u)} w_{(u, v)}} \textcolor{blue}{}$
\EndFor
\State $r_s(u) = 0 $
\EndProcedure
\end{algorithmic}
\label{algo:forward-local-push-weighted}
\end{algorithm}

%% file: main.bbl
%%% -*-BibTeX-*-
%%% Do NOT edit. File created by BibTeX with style
%%% ACM-Reference-Format-Journals [18-Jan-2012].

\begin{thebibliography}{41}

%%% ====================================================================
%%% NOTE TO THE USER: you can override these defaults by providing
%%% customized versions of any of these macros before the \bibliography
%%% command.  Each of them MUST provide its own final punctuation,
%%% except for \shownote{}, \showDOI{}, and \showURL{}.  The latter two
%%% do not use final punctuation, in order to avoid confusing it with
%%% the Web address.
%%%
%%% To suppress output of a particular field, define its macro to expand
%%% to an empty string, or better, \unskip, like this:
%%%
%%% \newcommand{\showDOI}[1]{\unskip}   % LaTeX syntax
%%%
%%% \def \showDOI #1{\unskip}           % plain TeX syntax
%%%
%%% ====================================================================

\ifx \showCODEN    \undefined \def \showCODEN     #1{\unskip}     \fi
\ifx \showDOI      \undefined \def \showDOI       #1{#1}\fi
\ifx \showISBNx    \undefined \def \showISBNx     #1{\unskip}     \fi
\ifx \showISBNxiii \undefined \def \showISBNxiii  #1{\unskip}     \fi
\ifx \showISSN     \undefined \def \showISSN      #1{\unskip}     \fi
\ifx \showLCCN     \undefined \def \showLCCN      #1{\unskip}     \fi
\ifx \shownote     \undefined \def \shownote      #1{#1}          \fi
\ifx \showarticletitle \undefined \def \showarticletitle #1{#1}   \fi
\ifx \showURL      \undefined \def \showURL       {\relax}        \fi
% The following commands are used for tagged output and should be
% invisible to TeX
\providecommand\bibfield[2]{#2}
\providecommand\bibinfo[2]{#2}
\providecommand\natexlab[1]{#1}
\providecommand\showeprint[2][]{arXiv:#2}

\bibitem[\protect\citeauthoryear{Aggarwal, Zhao, and Philip}{Aggarwal
  et~al\mbox{.}}{2011}]%
        {aggarwal2011outlier}
\bibfield{author}{\bibinfo{person}{Charu~C Aggarwal}, \bibinfo{person}{Yuchen
  Zhao}, {and} \bibinfo{person}{S~Yu Philip}.} \bibinfo{year}{2011}\natexlab{}.
\newblock \showarticletitle{Outlier detection in graph streams}. In
  \bibinfo{booktitle}{\emph{International Conference on Data Engineering}}.
  IEEE, \bibinfo{pages}{399--409}.
\newblock


\bibitem[\protect\citeauthoryear{Akoglu, McGlohon, and Faloutsos}{Akoglu
  et~al\mbox{.}}{2010}]%
        {akoglu2010oddball}
\bibfield{author}{\bibinfo{person}{Leman Akoglu}, \bibinfo{person}{Mary
  McGlohon}, {and} \bibinfo{person}{Christos Faloutsos}.}
  \bibinfo{year}{2010}\natexlab{}.
\newblock \showarticletitle{Oddball: Spotting anomalies in weighted graphs}. In
  \bibinfo{booktitle}{\emph{Pacific-Asia conference on knowledge discovery and
  data mining}}. Springer, \bibinfo{pages}{410--421}.
\newblock


\bibitem[\protect\citeauthoryear{Andersen, Chung, and Lang}{Andersen
  et~al\mbox{.}}{2006}]%
        {andersen2006local}
\bibfield{author}{\bibinfo{person}{Reid Andersen}, \bibinfo{person}{Fan Chung},
  {and} \bibinfo{person}{Kevin Lang}.} \bibinfo{year}{2006}\natexlab{}.
\newblock \showarticletitle{Local graph partitioning using pagerank vectors}.
  In \bibinfo{booktitle}{\emph{The IEEE Symposium on Foundations of Computer
  Science (FOCS)}}. IEEE, \bibinfo{pages}{475--486}.
\newblock


\bibitem[\protect\citeauthoryear{Berkhin}{Berkhin}{2006}]%
        {berkhin2006bookmark}
\bibfield{author}{\bibinfo{person}{Pavel Berkhin}.}
  \bibinfo{year}{2006}\natexlab{}.
\newblock \showarticletitle{Bookmark-coloring algorithm for personalized
  pagerank computing}.
\newblock \bibinfo{journal}{\emph{Internet Mathematics}} \bibinfo{volume}{3},
  \bibinfo{number}{1} (\bibinfo{year}{2006}), \bibinfo{pages}{41--62}.
\newblock


\bibitem[\protect\citeauthoryear{Beutel, Xu, Guruswami, Palow, and
  Faloutsos}{Beutel et~al\mbox{.}}{2013}]%
        {beutel2013copycatch}
\bibfield{author}{\bibinfo{person}{Alex Beutel}, \bibinfo{person}{Wanhong Xu},
  \bibinfo{person}{Venkatesan Guruswami}, \bibinfo{person}{Christopher Palow},
  {and} \bibinfo{person}{Christos Faloutsos}.} \bibinfo{year}{2013}\natexlab{}.
\newblock \showarticletitle{Copycatch: stopping group attacks by spotting
  lockstep behavior in social networks}. In
  \bibinfo{booktitle}{\emph{International World Wide Web Conference}}.
  \bibinfo{pages}{119--130}.
\newblock


\bibitem[\protect\citeauthoryear{Bhatia, Hooi, Yoon, Shin, and
  Faloutsos}{Bhatia et~al\mbox{.}}{2020}]%
        {bhatia2020midas}
\bibfield{author}{\bibinfo{person}{Siddharth Bhatia}, \bibinfo{person}{Bryan
  Hooi}, \bibinfo{person}{Minji Yoon}, \bibinfo{person}{Kijung Shin}, {and}
  \bibinfo{person}{Christos Faloutsos}.} \bibinfo{year}{2020}\natexlab{}.
\newblock \showarticletitle{Midas: Microcluster-based detector of anomalies in
  edge streams}. In \bibinfo{booktitle}{\emph{Proceedings of the AAAI
  conference on artificial intelligence}}, Vol.~\bibinfo{volume}{34}.
  \bibinfo{pages}{3242--3249}.
\newblock


\bibitem[\protect\citeauthoryear{Bhatia, Jain, Li, Kumar, and Hooi}{Bhatia
  et~al\mbox{.}}{2021a}]%
        {bhatia2021mstream}
\bibfield{author}{\bibinfo{person}{Siddharth Bhatia}, \bibinfo{person}{Arjit
  Jain}, \bibinfo{person}{Pan Li}, \bibinfo{person}{Ritesh Kumar}, {and}
  \bibinfo{person}{Bryan Hooi}.} \bibinfo{year}{2021}\natexlab{a}.
\newblock \showarticletitle{MSTREAM: Fast Anomaly Detection in Multi-Aspect
  Streams}. In \bibinfo{booktitle}{\emph{Proceedings of the Web Conference
  2021}}. \bibinfo{pages}{3371--3382}.
\newblock


\bibitem[\protect\citeauthoryear{Bhatia, Wadhwa, Yu, and Hooi}{Bhatia
  et~al\mbox{.}}{2021b}]%
        {bhatia2021sketch}
\bibfield{author}{\bibinfo{person}{Siddharth Bhatia}, \bibinfo{person}{Mohit
  Wadhwa}, \bibinfo{person}{Philip~S Yu}, {and} \bibinfo{person}{Bryan Hooi}.}
  \bibinfo{year}{2021}\natexlab{b}.
\newblock \showarticletitle{Sketch-Based Streaming Anomaly Detection in Dynamic
  Graphs}.
\newblock \bibinfo{journal}{\emph{arXiv preprint arXiv:2106.04486}}
  (\bibinfo{year}{2021}).
\newblock


\bibitem[\protect\citeauthoryear{Chang, Li, Sosic, Afifi, Schweighauser, and
  Leskovec}{Chang et~al\mbox{.}}{2021}]%
        {chang2021f}
\bibfield{author}{\bibinfo{person}{Yen-Yu Chang}, \bibinfo{person}{Pan Li},
  \bibinfo{person}{Rok Sosic}, \bibinfo{person}{MH Afifi},
  \bibinfo{person}{Marco Schweighauser}, {and} \bibinfo{person}{Jure
  Leskovec}.} \bibinfo{year}{2021}\natexlab{}.
\newblock \showarticletitle{F-fade: Frequency factorization for anomaly
  detection in edge streams}. In \bibinfo{booktitle}{\emph{ACM International
  Conference on Web Search and Data Mining}}. \bibinfo{pages}{589--597}.
\newblock


\bibitem[\protect\citeauthoryear{Chen, Sultan, Tian, Chen, and Skiena}{Chen
  et~al\mbox{.}}{2019}]%
        {chen2019fast}
\bibfield{author}{\bibinfo{person}{Haochen Chen}, \bibinfo{person}{Syed~Fahad
  Sultan}, \bibinfo{person}{Yingtao Tian}, \bibinfo{person}{Muhao Chen}, {and}
  \bibinfo{person}{Steven Skiena}.} \bibinfo{year}{2019}\natexlab{}.
\newblock \showarticletitle{Fast and accurate network embeddings via very
  sparse random projection}. In \bibinfo{booktitle}{\emph{International
  Conference on Information and Knowledge Management}}.
  \bibinfo{pages}{399--408}.
\newblock


\bibitem[\protect\citeauthoryear{Cohen}{Cohen}{[n.d.]}]%
        {cohen2005enron}
\bibfield{author}{\bibinfo{person}{W.W. Cohen}.}
  \bibinfo{year}{[n.d.]}\natexlab{}.
\newblock \bibinfo{title}{Enron email dataset}.
\newblock
\newblock
\newblock
\shownote{http://www.cs.cmu.edu/~enron/. Accessed in 2009.}


\bibitem[\protect\citeauthoryear{Eswaran and Faloutsos}{Eswaran and
  Faloutsos}{2018}]%
        {eswaran2018sedanspot}
\bibfield{author}{\bibinfo{person}{Dhivya Eswaran} {and}
  \bibinfo{person}{Christos Faloutsos}.} \bibinfo{year}{2018}\natexlab{}.
\newblock \showarticletitle{Sedanspot: Detecting anomalies in edge streams}. In
  \bibinfo{booktitle}{\emph{IEEE International Conference on Data Mining
  (ICDM)}}. IEEE, \bibinfo{pages}{953--958}.
\newblock


\bibitem[\protect\citeauthoryear{Eswaran, Faloutsos, Guha, and Mishra}{Eswaran
  et~al\mbox{.}}{2018}]%
        {eswaran2018spotlight}
\bibfield{author}{\bibinfo{person}{Dhivya Eswaran}, \bibinfo{person}{Christos
  Faloutsos}, \bibinfo{person}{Sudipto Guha}, {and} \bibinfo{person}{Nina
  Mishra}.} \bibinfo{year}{2018}\natexlab{}.
\newblock \showarticletitle{Spotlight: Detecting anomalies in streaming
  graphs}. In \bibinfo{booktitle}{\emph{Proceedings of the 24th ACM SIGKDD
  Conference on Knowledge Discovery \& Data Mining}}.
  \bibinfo{pages}{1378--1386}.
\newblock


\bibitem[\protect\citeauthoryear{Gottschalk and Demidova}{Gottschalk and
  Demidova}{2018}]%
        {gottschalk2018eventkg}
\bibfield{author}{\bibinfo{person}{Simon Gottschalk} {and}
  \bibinfo{person}{Elena Demidova}.} \bibinfo{year}{2018}\natexlab{}.
\newblock \showarticletitle{EventKG: A Multilingual Event-Centric Temporal
  Knowledge Graph}. In \bibinfo{booktitle}{\emph{Proceedings of the Extended
  Semantic Web Conference (ESWC 2018)}}. Springer.
\newblock


\bibitem[\protect\citeauthoryear{Gottschalk and Demidova}{Gottschalk and
  Demidova}{2019}]%
        {gottschalk2019eventkg}
\bibfield{author}{\bibinfo{person}{Simon Gottschalk} {and}
  \bibinfo{person}{Elena Demidova}.} \bibinfo{year}{2019}\natexlab{}.
\newblock \showarticletitle{{EventKG - the Hub of Event Knowledge on the Web -
  and Biographical Timeline Generation}}.
\newblock \bibinfo{journal}{\emph{Semantic Web Journal (SWJ)}}
  \bibinfo{volume}{10}, \bibinfo{number}{6}, \bibinfo{pages}{1039--1070}.
\newblock


\bibitem[\protect\citeauthoryear{Gower}{Gower}{1975}]%
        {gower1975generalized}
\bibfield{author}{\bibinfo{person}{John~C Gower}.}
  \bibinfo{year}{1975}\natexlab{}.
\newblock \showarticletitle{Generalized procrustes analysis}.
\newblock \bibinfo{journal}{\emph{Psychometrika}} \bibinfo{volume}{40},
  \bibinfo{number}{1} (\bibinfo{year}{1975}), \bibinfo{pages}{33--51}.
\newblock


\bibitem[\protect\citeauthoryear{Goyal, Kamra, He, and Liu}{Goyal
  et~al\mbox{.}}{2018}]%
        {goyal2018dyngem}
\bibfield{author}{\bibinfo{person}{Palash Goyal}, \bibinfo{person}{Nitin
  Kamra}, \bibinfo{person}{Xinran He}, {and} \bibinfo{person}{Yan Liu}.}
  \bibinfo{year}{2018}\natexlab{}.
\newblock \showarticletitle{Dyngem: Deep embedding method for dynamic graphs}.
\newblock \bibinfo{journal}{\emph{arXiv preprint arXiv:1805.11273}}
  (\bibinfo{year}{2018}).
\newblock


\bibitem[\protect\citeauthoryear{Guo, Li, Sha, and Tan}{Guo
  et~al\mbox{.}}{2017}]%
        {guo2017parallel}
\bibfield{author}{\bibinfo{person}{Wentian Guo}, \bibinfo{person}{Yuchen Li},
  \bibinfo{person}{Mo Sha}, {and} \bibinfo{person}{Kian-Lee Tan}.}
  \bibinfo{year}{2017}\natexlab{}.
\newblock \showarticletitle{Parallel personalized pagerank on dynamic graphs}.
\newblock \bibinfo{journal}{\emph{VLDB Endowment}} \bibinfo{volume}{11},
  \bibinfo{number}{1} (\bibinfo{year}{2017}), \bibinfo{pages}{93--106}.
\newblock


\bibitem[\protect\citeauthoryear{Guo, Zhou, and Skiena}{Guo
  et~al\mbox{.}}{2021}]%
        {xingzhi2021subset}
\bibfield{author}{\bibinfo{person}{Xingzhi Guo}, \bibinfo{person}{Baojian
  Zhou}, {and} \bibinfo{person}{Steven Skiena}.}
  \bibinfo{year}{2021}\natexlab{}.
\newblock \showarticletitle{Subset Node Representation Learning over Large
  Dynamic Graphs}. In \bibinfo{booktitle}{\emph{Proceedings of the 27th ACM
  SIGKDD Conference on Knowledge Discovery \& Data Mining}}.
  \bibinfo{pages}{516–526}.
\newblock


\bibitem[\protect\citeauthoryear{Kumar, Zhang, and Leskovec}{Kumar
  et~al\mbox{.}}{2019}]%
        {kumar2019predicting}
\bibfield{author}{\bibinfo{person}{Srijan Kumar}, \bibinfo{person}{Xikun
  Zhang}, {and} \bibinfo{person}{Jure Leskovec}.}
  \bibinfo{year}{2019}\natexlab{}.
\newblock \showarticletitle{Predicting dynamic embedding trajectory in temporal
  interaction networks}. In \bibinfo{booktitle}{\emph{Proceedings of the 25th
  ACM SIGKDD Conference on Knowledge Discovery \& Data Mining}}.
  \bibinfo{pages}{1269--1278}.
\newblock


\bibitem[\protect\citeauthoryear{Lippmann, Haines, Fried, Korba, and
  Das}{Lippmann et~al\mbox{.}}{2000}]%
        {lippmann20001999}
\bibfield{author}{\bibinfo{person}{Richard Lippmann}, \bibinfo{person}{Joshua~W
  Haines}, \bibinfo{person}{David~J Fried}, \bibinfo{person}{Jonathan Korba},
  {and} \bibinfo{person}{Kumar Das}.} \bibinfo{year}{2000}\natexlab{}.
\newblock \showarticletitle{The 1999 DARPA off-line intrusion detection
  evaluation}.
\newblock \bibinfo{journal}{\emph{Computer networks}} \bibinfo{volume}{34},
  \bibinfo{number}{4} (\bibinfo{year}{2000}), \bibinfo{pages}{579--595}.
\newblock


\bibitem[\protect\citeauthoryear{Lu, Wang, Shi, Yu, and Ye}{Lu
  et~al\mbox{.}}{2019}]%
        {lu2019temporal}
\bibfield{author}{\bibinfo{person}{Yuanfu Lu}, \bibinfo{person}{Xiao Wang},
  \bibinfo{person}{Chuan Shi}, \bibinfo{person}{Philip~S Yu}, {and}
  \bibinfo{person}{Yanfang Ye}.} \bibinfo{year}{2019}\natexlab{}.
\newblock \showarticletitle{Temporal network embedding with micro-and
  macro-dynamics}. In \bibinfo{booktitle}{\emph{International Conference on
  Information and Knowledge Management}}. \bibinfo{pages}{469--478}.
\newblock


\bibitem[\protect\citeauthoryear{Page, Brin, Motwani, and Winograd}{Page
  et~al\mbox{.}}{1999}]%
        {page1999pagerank}
\bibfield{author}{\bibinfo{person}{Lawrence Page}, \bibinfo{person}{Sergey
  Brin}, \bibinfo{person}{Rajeev Motwani}, {and} \bibinfo{person}{Terry
  Winograd}.} \bibinfo{year}{1999}\natexlab{}.
\newblock \bibinfo{booktitle}{\emph{The PageRank citation ranking: Bringing
  order to the web.}}
\newblock \bibinfo{type}{{T}echnical {R}eport}. \bibinfo{institution}{Stanford
  InfoLab}.
\newblock


\bibitem[\protect\citeauthoryear{Post{\u{a}}varu, Tsitsulin, de~Almeida, Tian,
  Lattanzi, and Perozzi}{Post{\u{a}}varu et~al\mbox{.}}{2020}]%
        {postuavaru2020instantembedding}
\bibfield{author}{\bibinfo{person}{{\c{S}}tefan Post{\u{a}}varu},
  \bibinfo{person}{Anton Tsitsulin}, \bibinfo{person}{Filipe
  Miguel~Gon{\c{c}}alves de Almeida}, \bibinfo{person}{Yingtao Tian},
  \bibinfo{person}{Silvio Lattanzi}, {and} \bibinfo{person}{Bryan Perozzi}.}
  \bibinfo{year}{2020}\natexlab{}.
\newblock \showarticletitle{InstantEmbedding: Efficient Local Node
  Representations}.
\newblock \bibinfo{journal}{\emph{arXiv preprint arXiv:2010.06992}}
  (\bibinfo{year}{2020}).
\newblock


\bibitem[\protect\citeauthoryear{Qiu, Dong, Ma, Li, Wang, Wang, and Tang}{Qiu
  et~al\mbox{.}}{2019}]%
        {qiu2019netsmf}
\bibfield{author}{\bibinfo{person}{Jiezhong Qiu}, \bibinfo{person}{Yuxiao
  Dong}, \bibinfo{person}{Hao Ma}, \bibinfo{person}{Jian Li},
  \bibinfo{person}{Chi Wang}, \bibinfo{person}{Kuansan Wang}, {and}
  \bibinfo{person}{Jie Tang}.} \bibinfo{year}{2019}\natexlab{}.
\newblock \showarticletitle{Netsmf: Large-scale network embedding as sparse
  matrix factorization}. In \bibinfo{booktitle}{\emph{International World Wide
  Web Conference}}. \bibinfo{pages}{1509--1520}.
\newblock


\bibitem[\protect\citeauthoryear{Ranshous, Harenberg, Sharma, and
  Samatova}{Ranshous et~al\mbox{.}}{2016}]%
        {ranshous2016scalable}
\bibfield{author}{\bibinfo{person}{Stephen Ranshous}, \bibinfo{person}{Steve
  Harenberg}, \bibinfo{person}{Kshitij Sharma}, {and} \bibinfo{person}{Nagiza~F
  Samatova}.} \bibinfo{year}{2016}\natexlab{}.
\newblock \showarticletitle{A scalable approach for outlier detection in edge
  streams using sketch-based approximations}. In \bibinfo{booktitle}{\emph{IEEE
  International Conference on Data Mining (ICDM)}}. SIAM,
  \bibinfo{pages}{189--197}.
\newblock


\bibitem[\protect\citeauthoryear{Rossi and Ahmed}{Rossi and Ahmed}{2015}]%
        {nr-aaai15}
\bibfield{author}{\bibinfo{person}{Ryan~A. Rossi} {and}
  \bibinfo{person}{Nesreen~K. Ahmed}.} \bibinfo{year}{2015}\natexlab{}.
\newblock \showarticletitle{The Network Data Repository with Interactive Graph
  Analytics and Visualization}. In \bibinfo{booktitle}{\emph{Proceedings of the
  AAAI conference on artificial intelligence}}.
\newblock


\bibitem[\protect\citeauthoryear{Sankar, Wu, Gou, Zhang, and Yang}{Sankar
  et~al\mbox{.}}{2020}]%
        {sankar2020dysat}
\bibfield{author}{\bibinfo{person}{Aravind Sankar}, \bibinfo{person}{Yanhong
  Wu}, \bibinfo{person}{Liang Gou}, \bibinfo{person}{Wei Zhang}, {and}
  \bibinfo{person}{Hao Yang}.} \bibinfo{year}{2020}\natexlab{}.
\newblock \showarticletitle{Dysat: Deep neural representation learning on
  dynamic graphs via self-attention networks}. In \bibinfo{booktitle}{\emph{ACM
  International Conference on Web Search and Data Mining}}.
  \bibinfo{pages}{519--527}.
\newblock


\bibitem[\protect\citeauthoryear{Shin, Hooi, Kim, and Faloutsos}{Shin
  et~al\mbox{.}}{2017}]%
        {shin2017densealert}
\bibfield{author}{\bibinfo{person}{Kijung Shin}, \bibinfo{person}{Bryan Hooi},
  \bibinfo{person}{Jisu Kim}, {and} \bibinfo{person}{Christos Faloutsos}.}
  \bibinfo{year}{2017}\natexlab{}.
\newblock \showarticletitle{Densealert: Incremental dense-subtensor detection
  in tensor streams}. In \bibinfo{booktitle}{\emph{ACM SIGKDD Conference on
  Knowledge Discovery \& Data Mining}}. \bibinfo{pages}{1057--1066}.
\newblock


\bibitem[\protect\citeauthoryear{Tsitsulin, Munkhoeva, Mottin, Karras,
  Oseledets, and M{\"u}ller}{Tsitsulin et~al\mbox{.}}{2021}]%
        {tsitsulin2021frede}
\bibfield{author}{\bibinfo{person}{Anton Tsitsulin}, \bibinfo{person}{Marina
  Munkhoeva}, \bibinfo{person}{Davide Mottin}, \bibinfo{person}{Panagiotis
  Karras}, \bibinfo{person}{Ivan Oseledets}, {and} \bibinfo{person}{Emmanuel
  M{\"u}ller}.} \bibinfo{year}{2021}\natexlab{}.
\newblock \showarticletitle{FREDE: anytime graph embeddings}.
\newblock \bibinfo{journal}{\emph{VLDB Endowment}} \bibinfo{volume}{14},
  \bibinfo{number}{6} (\bibinfo{year}{2021}), \bibinfo{pages}{1102--1110}.
\newblock


\bibitem[\protect\citeauthoryear{Wang and Mahadevan}{Wang and
  Mahadevan}{2008}]%
        {wang2008manifold}
\bibfield{author}{\bibinfo{person}{Chang Wang} {and} \bibinfo{person}{Sridhar
  Mahadevan}.} \bibinfo{year}{2008}\natexlab{}.
\newblock \showarticletitle{Manifold alignment using procrustes analysis}. In
  \bibinfo{booktitle}{\emph{International Conference on Machine Learning
  (ICML)}}. \bibinfo{pages}{1120--1127}.
\newblock


\bibitem[\protect\citeauthoryear{Wang, Yang, Xiao, Wei, and Yang}{Wang
  et~al\mbox{.}}{2017}]%
        {wang2017fora}
\bibfield{author}{\bibinfo{person}{Sibo Wang}, \bibinfo{person}{Renchi Yang},
  \bibinfo{person}{Xiaokui Xiao}, \bibinfo{person}{Zhewei Wei}, {and}
  \bibinfo{person}{Yin Yang}.} \bibinfo{year}{2017}\natexlab{}.
\newblock \showarticletitle{FORA: simple and effective approximate
  single-source personalized pagerank}. In \bibinfo{booktitle}{\emph{ACM SIGKDD
  Conference on Knowledge Discovery \& Data Mining}}.
  \bibinfo{pages}{505--514}.
\newblock


\bibitem[\protect\citeauthoryear{Wang, Fang, Lin, and Wu}{Wang
  et~al\mbox{.}}{2015}]%
        {wang2015localizing}
\bibfield{author}{\bibinfo{person}{Teng Wang}, \bibinfo{person}{Chunsheng
  Fang}, \bibinfo{person}{Derek Lin}, {and} \bibinfo{person}{S~Felix Wu}.}
  \bibinfo{year}{2015}\natexlab{}.
\newblock \showarticletitle{Localizing temporal anomalies in large evolving
  graphs}. In \bibinfo{booktitle}{\emph{IEEE International Conference on Data
  Mining (ICDM)}}. SIAM, \bibinfo{pages}{927--935}.
\newblock


\bibitem[\protect\citeauthoryear{Wei, He, Xiao, Wang, Shang, and Wen}{Wei
  et~al\mbox{.}}{2018}]%
        {wei2018topppr}
\bibfield{author}{\bibinfo{person}{Zhewei Wei}, \bibinfo{person}{Xiaodong He},
  \bibinfo{person}{Xiaokui Xiao}, \bibinfo{person}{Sibo Wang},
  \bibinfo{person}{Shuo Shang}, {and} \bibinfo{person}{Ji-Rong Wen}.}
  \bibinfo{year}{2018}\natexlab{}.
\newblock \showarticletitle{Topppr: top-k personalized pagerank queries with
  precision guarantees on large graphs}. In
  \bibinfo{booktitle}{\emph{International Conference on Information and
  Knowledge Management}}. \bibinfo{pages}{441--456}.
\newblock


\bibitem[\protect\citeauthoryear{Wu, Gan, Wei, and Zhang}{Wu
  et~al\mbox{.}}{2021}]%
        {wu2021unifying}
\bibfield{author}{\bibinfo{person}{Hao Wu}, \bibinfo{person}{Junhao Gan},
  \bibinfo{person}{Zhewei Wei}, {and} \bibinfo{person}{Rui Zhang}.}
  \bibinfo{year}{2021}\natexlab{}.
\newblock \showarticletitle{Unifying the Global and Local Approaches: An
  Efficient Power Iteration with Forward Push}. In
  \bibinfo{booktitle}{\emph{Proceedings of the 2021 International Conference on
  Management of Data}}. \bibinfo{pages}{1996--2008}.
\newblock


\bibitem[\protect\citeauthoryear{Yoon, Hooi, Shin, and Faloutsos}{Yoon
  et~al\mbox{.}}{2019}]%
        {yoon2019fast}
\bibfield{author}{\bibinfo{person}{Minji Yoon}, \bibinfo{person}{Bryan Hooi},
  \bibinfo{person}{Kijung Shin}, {and} \bibinfo{person}{Christos Faloutsos}.}
  \bibinfo{year}{2019}\natexlab{}.
\newblock \showarticletitle{Fast and accurate anomaly detection in dynamic
  graphs with a two-pronged approach}. In \bibinfo{booktitle}{\emph{Proceedings
  of the 25th ACM SIGKDD Conference on Knowledge Discovery \& Data Mining}}.
  \bibinfo{pages}{647--657}.
\newblock


\bibitem[\protect\citeauthoryear{Yoon, Jin, and Kang}{Yoon
  et~al\mbox{.}}{2018}]%
        {yoon2018fast}
\bibfield{author}{\bibinfo{person}{Minji Yoon}, \bibinfo{person}{Woojeong Jin},
  {and} \bibinfo{person}{U Kang}.} \bibinfo{year}{2018}\natexlab{}.
\newblock \showarticletitle{Fast and accurate random walk with restart on
  dynamic graphs with guarantees}. In \bibinfo{booktitle}{\emph{International
  World Wide Web Conference}}. \bibinfo{pages}{409--418}.
\newblock


\bibitem[\protect\citeauthoryear{Yu, Cheng, Aggarwal, Zhang, Chen, and Wang}{Yu
  et~al\mbox{.}}{2018}]%
        {yu2018netwalk}
\bibfield{author}{\bibinfo{person}{Wenchao Yu}, \bibinfo{person}{Wei Cheng},
  \bibinfo{person}{Charu~C Aggarwal}, \bibinfo{person}{Kai Zhang},
  \bibinfo{person}{Haifeng Chen}, {and} \bibinfo{person}{Wei Wang}.}
  \bibinfo{year}{2018}\natexlab{}.
\newblock \showarticletitle{Netwalk: A flexible deep embedding approach for
  anomaly detection in dynamic networks}. In
  \bibinfo{booktitle}{\emph{Proceedings of the 24th ACM SIGKDD Conference on
  Knowledge Discovery \& Data Mining}}. \bibinfo{pages}{2672--2681}.
\newblock


\bibitem[\protect\citeauthoryear{Zhang, Lofgren, and Goel}{Zhang
  et~al\mbox{.}}{2016}]%
        {zhang2016approximate}
\bibfield{author}{\bibinfo{person}{Hongyang Zhang}, \bibinfo{person}{Peter
  Lofgren}, {and} \bibinfo{person}{Ashish Goel}.}
  \bibinfo{year}{2016}\natexlab{}.
\newblock \showarticletitle{Approximate personalized pagerank on dynamic
  graphs}. In \bibinfo{booktitle}{\emph{Proceedings of the 22nd ACM SIGKDD
  Conference on Knowledge Discovery \& Data Mining}}.
  \bibinfo{pages}{1315--1324}.
\newblock


\bibitem[\protect\citeauthoryear{Zhang, Cui, Li, Wang, and Zhu}{Zhang
  et~al\mbox{.}}{2018}]%
        {zhang2018billion}
\bibfield{author}{\bibinfo{person}{Ziwei Zhang}, \bibinfo{person}{Peng Cui},
  \bibinfo{person}{Haoyang Li}, \bibinfo{person}{Xiao Wang}, {and}
  \bibinfo{person}{Wenwu Zhu}.} \bibinfo{year}{2018}\natexlab{}.
\newblock \showarticletitle{Billion-scale network embedding with iterative
  random projection}. In \bibinfo{booktitle}{\emph{IEEE International
  Conference on Data Mining (ICDM)}}. IEEE, \bibinfo{pages}{787--796}.
\newblock


\bibitem[\protect\citeauthoryear{Zhou, Yang, Ren, Wu, and Zhuang}{Zhou
  et~al\mbox{.}}{2018}]%
        {zhou2018dynamic}
\bibfield{author}{\bibinfo{person}{Lekui Zhou}, \bibinfo{person}{Yang Yang},
  \bibinfo{person}{Xiang Ren}, \bibinfo{person}{Fei Wu}, {and}
  \bibinfo{person}{Yueting Zhuang}.} \bibinfo{year}{2018}\natexlab{}.
\newblock \showarticletitle{Dynamic network embedding by modeling triadic
  closure process}. In \bibinfo{booktitle}{\emph{Proceedings of the AAAI
  conference on artificial intelligence}}, Vol.~\bibinfo{volume}{32}.
\newblock


\end{thebibliography}
